\documentclass[10pt]{amsart}
\usepackage{amsfonts,amssymb}
\usepackage{amsmath}
\usepackage{color}
\usepackage{pstcol,pst-node}
\usepackage{pst-plot}
\usepackage{graphics}
\usepackage{epic}
\usepackage{curves}

\newcommand{\tr}{\,{\rm Tr}\,}
\newcommand{\ID}{1\!\!1}

\def\eps{\varepsilon}
\def\be{\begin{equation}}
\def\ee{\end{equation}}
\def\bea{\begin{eqnarray}}
\def\eea{\end{eqnarray}}
\def\<{\langle}
\def\>{\rangle}
\def\ddz{\frac{\rm d}{{\rm d}z}}
\def\ddl{\frac{\rm d}{{\rm d}\lambda}}
\def\nn{\nonumber}

\def\Tr{{\rm Tr}}
\def\one#1{#1^{\raise8pt\hbox{$\scriptstyle\!\!\!\!1$}}\,{}}
\def\two#1{#1^{\raise8pt\hbox{$\scriptstyle\!\!\!\!2$}}\,{}}
\def\onetwo#1{#1^{\raise6pt\hbox{$\scriptstyle\!\!\!\!\!{12}$}}\,{}}
\def\otim{\mathop{\otimes}}

\newtheorem{theorem}{Theorem}[section]
\newtheorem{lm}[theorem]{Lemma}

\newtheorem{remark}[theorem]{Remark}
\theoremstyle{definition}

\newtheorem{example}[theorem]{Example}

\theoremstyle{remark}

\begin{document}

\title[Teichmuller spaces as degenerated symplectic leaves]
{Teichm\"uller spaces as degenerated symplectic leaves in
Dubrovin--Ugaglia Poisson manifolds}

\author{Leonid Chekhov$^{\ast}$}\thanks{$^{\ast}$Steklov Mathematical Institute, ITEP, and  Laboratoire Poncelet,
Moscow, Russia.}
\author{Marta Mazzocco$^\star$}\thanks{$^\star$School of Mathematica, Loughborough University, UK}

\maketitle

\vspace{-7cm}

\begin{center}
\hfill ITEP/TH-15/11
\end{center}

\vspace{7cm}

\rightline{\it To Boris Dubrovin in occasion of his sixtieth birthday.}

\begin{abstract}
In this paper we study the Goldman
bracket between geodesic length functions both on a Riemann
surface $\Sigma_{g,s,0}$ of genus $g$ with $s=1,2$ holes and on a
Riemann sphere $\Sigma_{0,1,n}$ with one hole and $n$ orbifold
points of order two. We show that the corresponding Teichm\"uller spaces
 $\mathcal T_{g,s,0}$  and $\mathcal T_{0,1,n}$ are realised as real slices of
 degenerated symplectic leaves  in the Dubrovin--Ugaglia Poisson algebra  of
upper--triangular matrices $S$ with $1$ on the diagonal.
\end{abstract}

\section{Introduction}

In this paper we study some special symplectic leaves in the
Poisson algebra (\ref{poissonu}) of upper--triangular matrices $S$
with $1$ on the diagonal. This algebra appears as the
semi--classical limit of  the famous Nelson--Regge algebra
in $2+1$-dimensional quantum gravity \cite{NR,NRZ},  and in
Chern--Simons theory as Fock--Rosly bracket \cite{FR}. At classical
level, this algebra was discovered in the context of Frobenius
manifold theory by Dubrovin and Ugaglia \cite{Dub7,Ugaglia} and in
the study of non--symmetric bilinear forms by Bondal \cite{Bondal}.

In this paper we adopt the isomonodromic deformations
perspective. According to Dubrovin's isomonodromicity theorem part
III  \cite{Dub7},  the metric, the flat coordinates, the
pre--potential and the structure constants of a $n$ dimensional
semi--simple Frobenius manifold are given by the space of parameters
${\bf u}=(u_1,\dots,u_n)$ together with an $n\times n$
skew-symmetric matrix function $V({\bf u})$ such that the linear
differential operator
$$
\Lambda(z):= \ddz -U -\frac{V({\bf u})}{z},\qquad U=\rm{diagonal}({\bf u}),
$$
has constant
monodromy data as $(u_1,\dots,u_n)$ vary in the configuration space of $n$ points. Generically, the monodromy data of
$\Lambda(z)$ are encoded in the so-called {\it Stokes matrix}\/
$S$, an upper triangular matrix with $1$ on the diagonal.

It turns out that, although the monodromy map
$$
V(u) \to S,
$$
is given by complicated transcendent functions, the Poisson bracket on the space of Stokes matrices is given by very simple  quadratic formulae:
\begin{eqnarray}\label{poissonu}
&&
\left\{s_{ik},s_{jl}\right\}=0,\quad\hbox{for}\,i<k<j<l,\nn
\\&&
\left\{s_{ik},s_{jl}\right\}=0,\quad\hbox{for}\, i<j<l<k,\nn
\\&&
\left\{s_{ik},s_{jl}\right\}={i\pi}\left(s_{ij}s_{kl}-s_{il}s_{kj}\right),\quad\hbox{for}\, i<j<k<l,
\\&&
\left\{s_{ik},s_{kl}\right\}=\frac{i\pi}{2}\left(s_{ik}s_{kl}-2s_{il}\right),\quad\hbox{for}\, i<k<l,\nn
\\&&
\left\{s_{ik},s_{jk}\right\}=-\frac{i \pi}{2}\left(s_{ik}s_{jk}-2s_{ij}\right),\quad\hbox{for} \, i< j<k,\nn
\\&&
\left\{s_{ik},s_{il}\right\}=-\frac{i\pi}{2}\left(s_{ik}s_{il}-2s_{kl}\right),\quad\hbox{for} \, i<k<l.\nn
\end{eqnarray}
This bracket was obtained in \cite{Dub7} in the case $n=3$, then for
any $n>3$ in \cite{Ugaglia}, and for this reason it is called the {\it
Dubrovin--Ugaglia} bracket.

The same bracket appeared in Teichm\"uller theory as the Goldman
bracket \cite{Gold} between geodesic length functions both on a Riemann
surface $\Sigma_{g,s,0}$ of genus $g$ with $s=1,2$ holes and on a
Riemann sphere $\Sigma_{0,1,n}$ with one hole and $n$ orbifold
points of order two. Let us denote the two Teichm\"uller spaces by
$\mathcal T_{g,s,0}$, $s=1,2$ and $\mathcal T_{0,1,n}$ respectively.
These are {\it real symplectic manifolds} of dimension respectively
$$
\dim_{\mathbb R}\left(\mathcal T_{g,s,0}
\right)= \left\{\begin{array}{ll}
3n-7 &\hbox{for $n$ odd}\\
3n-8 &\hbox{for $n$ even,}\\
\end{array}\right.\,\hbox{with}\,\, g=\left[\frac{n-1}{2}\right], \,
s=\left\{\begin{array}{ll}
1&\hbox{for $n$ odd,}\\
2&\hbox{for $n$ even,}\\
\end{array}\right.
$$
and
$$
\dim_{\mathbb R}\left(\mathcal T_{0,1,n}\right)= 2(n-2),
$$
while the generic symplectic leaves $\mathcal L_{\small{generic}}$ in the Dubrovin--Ugaglia bracket have dimension
$$
\dim_{\mathbb C}\left(\mathcal L_{{generic}}\right)=\frac{n(n-1)}{2}-\left[\frac{n}{2}\right].
$$
It is natural to ask whether the Teichm\"uller spaces arise as real
slices of some subvarieties of a generic leaf or of a degenerated
leaf. In this paper, we prove that in the both cases the Teichm\"uller
spaces correspond to degenerated symplectic leaves whose complex dimension
is equal to the real dimension of the Teichm\"uller space itself (see Theorem \ref{th:main0}).
As a consequence, we
give flat coordinates on such degenerated symplectic leaves by introducing a suitable complexification of the shear coordinates.

The paper is organised as follows. In Section \ref{se:DU} we recall
some  facts about the isomonodromic deformations of the
operator $\Lambda(z)$ and about the Dubrovin--Ugaglia bracket. This part
is mainly a review, apart perhaps the minor Remark \ref{rm:nu}. In
Section \ref{se:geo}, we review some basics on Teichm\"uller theory
recalling the characterization of the Stokes matrices whose entries
arise as geodesic length functions  on a Riemann surface $\Sigma_{g,s,0}$ of
genus $g$ with $s=1,2$ holes and on a Riemann sphere
$\Sigma_{0,1,n}$ with one hole and $n$ orbifold points of order two.
Section \ref{se:bondal} is original and contains the
characterization of the symplectic leaves arising in Teichm\"uller
theory, including the proof of Theorem \ref{th:main0} stating
that in the both cases (i.e., for the  Riemann
surface $\Sigma_{g,s,0}$ of genus $g$ with $s=1,2$ holes and Riemann
sphere $\Sigma_{0,1,n}$ with one hole and $n$ orbifold points of
order two) the Teichm\"uller spaces $\mathcal T_{g,s,0}$  and
$\mathcal T_{0,1,n}$  are the real slices of  degenerated symplectic leaves
complex  dimension equal to the real dimension of the Teichm\"uller
space itself. In subsection \ref{ss:Min} we discuss an interesting
interpretation in terms of a $n$ particle model in Minkowski space
and in subsection \ref{suse:c} we discuss the complexification of
the shear coordinates. Section \ref{se:con} contains a  heuristic
discussion about some minor progress towards the characterisation of
the  Frobenius manifold structure on the Teichm\"uller spaces
$\mathcal T_{g,s,0}$  and $\mathcal T_{0,1,n}$.

\vskip 2mm \noindent{\bf Acknowledgements.}  The authors are
grateful to Boris Dubrovin, who put them in contact and gave them
many helpful suggestions. We would like to thank also J\"orgen Andersen,
Alexei Bondal, Bob Penner and Vasilisa Shramchenko for many
enlighting conversations. This research  was supported by EPSRC  ARF
EP/D071895/, by the Marie Curie training network ENIGMA, by the
ESF network MISGAM, by the Russian Foundation for Basic Research under
the grant No.~10-01-92104-YaF\_a, 11-02-90453-Ukr\_a, and~10-02-01315-a, by
the Ministry of Education and Science of the Russian Federation under contract
02.740.11.0608, by the Program of Supporting
Leading Scientific Schools No.~NSh-8265.2010.1, and by the Scientific Program
Mathematical Methods of Nonlinear Dynamics.

\section{Dubrovin--Ugaglia bracket}\label{se:DU}

In this Section, we recall some  facts about the monodromy data of
the operator $\Lambda(z)$,  its monodromy preserving deformations,
and the construction of the Dubrovin--Ugaglia bracket by the
so--called duality \cite{Dub8} which allows to map $\Lambda(z)$ to a
Fuchsian differential operator.

The Dubrovin--Ugaglia bracket is a Poisson bracket on the group of
upper--triangular matrices $S$ with $1$ on the diagonal. These
matrices $S$ arise as {\it monodromy data}\/ of the following system
of $n$ first order ODEs
\begin{equation}
\ddz Y= \left(U+\frac{V}{z}\right)Y
\label{irreg}\end{equation}
where $U={\rm diagonal}(u_1,\dots,u_n)$, $(u_1,\dots,u_n)\in X_n$,
$$
X_n=\left\{(u_1,\dots, u_n)\in\mathbb C^n\, | u_i\neq u_j \, \hbox{for}\, i\neq j
\right\}
$$
and $V=-V^{T}$ is a skew symmetric $n\times n$ matrix with eigenvalues $\mu_1,\dots,\mu_n$.

\subsection{Monodromy data}\label{se:mon-data}

A general description of monodromy data of linear systems of ODE can be
found in \cite{MJ1,MJ2,MJU}. Here we use the same notations as in \cite{Dub7}, where most results of this sub--section are proved.

We fix a real number $\varphi\in[0,2\pi[$ and consider the open subset ${\mathcal U}\in X_n$ such that
the rays $L_1,...,L_n$ defined by
\begin{equation}\label{eq:rays}
L_j:=\{u_j+i\rho e^{-i\varphi}\,|\, 0\leq\rho<\infty\}
\end{equation}
do not intersect. We assume that the points
$(u_1,\dots,u_n)\in{\mathcal U}$  are ordered in such a way that the
rays $L_1,\dots,L_n$ exit from infinity in counter-clockwise order.

In \cite{Dub7} it was proved that  for a fixed line $l$
$$
l:=\{\arg(z)=\varphi\},
$$
there exists $\eps>0$ small enough, $Z\in{\mathbb R}$ large enough,
two sectors $\Pi_L$ and $\Pi_R$ defined as
\begin{equation}
\begin{array}{cc}
&\Pi_R=\{z: \arg(l)-\pi-\eps<\arg(z)<\arg(l)+\eps,\, |z| > |Z|\}\\
&\Pi_L=\{z: \arg(l)-\eps<\arg(z)<\arg(l)+\pi+\eps,\, |z| > |Z|\}\\
\end{array}
\end{equation}
and two unique fundamental solutions $Y_L(z)$ in $\Pi_L$ and
$Y_R(z)$ in $\Pi_R$ such that
\begin{equation}
Y_{L,R}\sim\left(\ID+{\mathcal O}\left({1\over z}\right)\right)
 e^{z U},\quad\hbox{as }z\to\infty,\quad z\in \Pi_{L,R}.
\label{eq:fun-sol-LR}
\end{equation}
In the narrow sectors
$$
\begin{array}{cl}
\Pi_+:=&\{z|\, \varphi-\varepsilon< \arg z<\varphi+\varepsilon\}\\
\Pi_-:=&\{z|\,\varphi-\pi-\varepsilon< \arg z<\varphi-\pi+\varepsilon\}\\
\end{array}
$$
obtained by the intersection of $\Pi_L$ and $\Pi_R$, we have two
fundamental matrices with the same asymptotic behaviour
(\ref{eq:fun-sol-LR}). They are related by multiplication by a
constant invertible matrix
$$
Y_L (z) =Y_R (z) S_+, ~~z\in \Pi_+.
$$
$$
Y_L (z) =Y_R (z) S_-, ~~z\in \Pi_-.
$$
The matrices $S_+$ and $S_-$ are called {\it Stokes matrices.}\/ Due to the skew symmetry of $V$, they satisfy the following relation
$$
S_-^T=S_+:=S.
$$
Thanks to the choice of the order of $u_1,\dots,u_n$, $S$ is upper triangular with $1$ on the diagonal.

Near the regular singular point $0$, there exists a fundamental matrix of the system
(\ref{irreg}) of the form
\begin{equation}
Y_0(z)=\left({\Gamma}+{\mathcal O}(z)\right)z^{\mu} z^{{R}},\quad\hbox{as}\quad
 z\rightarrow 0,
\label{eq:fun-sol0}
\end{equation}
where a branch cut between zero and infinity has been fixed along the negative part $l_-$ of $l$, the matrix $\Gamma$ is the eigenvector matrix of $V$,
$V {\Gamma}={\Gamma} \mu$ and $R$ is a
nilpotent matrix satisfying the following relation:
\begin{equation}\label{lev1}
e^{2\pi i \mu} R = R\, e^{2 \pi i \mu}.
\end{equation}
The monodromy ${\mathcal M}_0$ of the system (\ref{irreg}) with respect to
the normalized fundamental matrix (\ref{eq:fun-sol0}) generated by a simple
closed loop around the origin is
$$
{M}_0=\exp(2\pi i\mu)\exp(2\pi i{R}).
$$
The {\it central
connection matrix}\/  $C$ between $0$ and $\infty$ is defined by
$$
Y_0(z) =Y_{L}(z)C, ~~z\in \Pi_{L}.
$$
The {\it monodromy data}\/ of the system (\ref{irreg}) consist of $(\mu,R,C,S)$ and are related by
\begin{equation}\label{eq:mon-data}
C^{-1}S^{-T}S C=exp(2\pi i\mu)\exp(2\pi i{R}).
\end{equation}

\subsection{Dual Fuchsian system and its monodromy data}\label{se:duality}
Following \cite{Dub8}, we consider a $n\times n$  Fuchsian system of the form
\begin{equation}
\ddl \Phi = \sum_{k=1}^n \frac{A_k}{\lambda-u_k}\Phi,
\label{eq:fuchs}
\end{equation}
where
\begin{equation}\label{eq:A-V}
A_k=E_k(\nu-\frac{1}{2}- V),
\end{equation}
and $\nu$ is an arbitrary parameter.
This system is dubbed {\it dual}\/  to the system (\ref{irreg}).
Let us remind how the monodromy data of this system (\ref{eq:fuchs}) are related to the monodromy data of system  (\ref{irreg}):

\begin{theorem}\label{th:q-duality}\cite{Dub8}
Let $q=e^{2\pi i \nu}$ and assume that $q$ is not a root of the characteristic equation
\begin{equation}\label{eq:det-condq}
\det\left(q S+S^T
\right)=0
\end{equation}
Then there exist $n$ linearly independent solutions
$\phi^{(1)},\dots,\phi^{(n)}$ of the system (\ref{eq:fuchs})
analytic in $\lambda\in\mathbb C\setminus \cup_j L_j$ such that the
monodromy transformations $M_1,\dots,M_n$ along the small loops
encircling counter--clockwise the points $u_1,\dots,u_n$ are given
by
\begin{equation}\label{eq:mon}
M_k=\ID-E_k (q S+S^T).
\end{equation}
The monodromy around infinity is given by $M_\infty=-\frac{1}{q} S^{-1} S^{T}$.
\end{theorem}

\subsection{Monodromy preserving deformations}

The monodromy preserving deformations equations for the system (\ref{irreg}) are the following non--linear differential equations
\begin{equation}
\frac{\partial V}{\partial u_i}=[V_i,V],\qquad
V_i={\rm ad}_{E_i}{\rm ad}_{U}^{-1}(V),\quad i=1,\dots,n,
\label{15}
\end{equation}
where $E_i$ is the matrix with entries
$E_{i_{kl}}=\delta_{ik}\delta_{il}$. For any solutions $V(u)$ of
equation (\ref{15}), the monodromy data $(\mu,R,C,S)$ of the system
$$
\ddz Y=\left(U+\frac{V(u)}{z}
\right)Y
$$
are constant in a disk in $X_n$. These same equations describe the
isomonodromic deformations of (\ref{eq:fuchs}), namely  the
monodromy data $M_1,\dots,M_n$ of the system
\begin{equation}
\ddl \Phi = \sum_{k=1}^n \frac{A_k(u)}{\lambda-u_k}\Phi,
\qquad A_k(u)=E_k(\nu-\frac{1}{2}- V(u)), \quad k=1,\dots,n,
\label{eq:fuchs-iso}
\end{equation}
are constant in a disk in $X_n$. Indeed equations (\ref{15}) are
equivalent to the Schlesinger equations \cite{Sch} for $A_1,\dots,A_n$. In
\cite{Har} it was proved that the spectral curve of these two
systems is the same.

The set of equations (\ref{15}) can be written as a $n$--times
Hamiltonian system on the space of skew--symmetric matrices $V$
equipped with the standard linear Poisson bracket for
$\mathfrak{so}(n)\ni V$:
\begin{equation}\label{poissonV}
\left\{ V_{ab},V_{cd}\right\} = V_{ad}\delta_{bc} + V_{bc}\delta_{ad} - V_{bd}\delta_{ac} + V_{ac}\delta_{bd}.
\end{equation}
Indeed equation (\ref{15}) can be rewritten as
\begin{equation}\label{hamform}
\frac{\partial V}{\partial u_i}=\left\{V,H_i\right\},
\end{equation}
where the Hamiltonian functions $H_i$ depend on the times $u_1,\dots,u_n$
\begin{equation}
H_i=\frac{1}{2}\sum_{j\neq i} \frac{V_{ij}^2}{u_i-u_j}.
\label{hamV}\end{equation}

Equivalently the isomonodromic deformations equations for  $A_1,\dots,A_n$ can be written as
\begin{equation}\label{hamformfuchs}
\frac{\partial A_k}{\partial u_i}=\left\{A_k,H_i\right\},
\end{equation}
where $\{\cdot,\cdot\}$ are the standard linear Poisson bracket for
$\mathfrak{gl}(n)\ni A_k$ and the Hamiltonian functions $H_i$ are
given by:
\begin{equation}
H_i=\frac{1}{2}\sum_{j\neq i} \frac{\Tr(A_i A_j)}{u_i-u_j},
\label{hamA}\end{equation}
and coincide with the previous ones thanks to the fact that
$$
\Tr(A_i A_j)=V_{ij}^2.
$$

\subsection{Korotkin--Samtleben bracket}

In this subsection we recall the definition of the
Korotkin--Samtleben bracket and obtain the Dubrovin--Ugaglia
bracket as its reduction when (\ref{eq:det-condq}) is satisfied.

According to \cite{KS} the standard Lie--Poisson bracket on $\mathfrak{gl}(n,{\mathbb C})$ is mapped by the monodromy map to
\begin{eqnarray}\nn
&&
\left\{ M_i\otim_,M_i\right\}
=\two{M_i}\Omega\one{M_i} - \one{M_i} \Omega \two{M_i} \\
&&
\left\{ M_i\otim_,M_j\right\}
=  \one{M_i} \Omega \two{M_j} +\two{M_j}\Omega\one{M_i}-\Omega\one{M_i}\two{M_j}-\Omega\two{M_j}\one{M_i} ,\quad\hbox{for}\,\,i<j.
\end{eqnarray}
This bracket does not satisfy the Jacobi identity, but it
reduces to a Poisson bracket on the adjoint invariant objects,
i.e., on the traces of the matrices $M_1,\dots, M_n$ and their
products.

If $q$ is chosen in such a way that condition (\ref{eq:det-condq})
is satisfied,  the monodromy matrices of the dual Fuchsian system
have the form (\ref{eq:mon}) so that
\begin{equation}\label{eq:S-M1}
\Tr(M_i M_j)=n-2-2q+q S_{ij}^2,\qquad i<j,
\end{equation}
where $S_{ij}$ is the $ij$ entry in the Stokes matrix $S$.
As a consequence the entries of  the Stokes matrix $S$ are adjoint
invariant, and the Korotkin--Samtleben bracket reduces to a Poisson bracket on them. This was precisely the main idea by
Ugaglia, she assumed $q=1$ and proved that for
$$
\det(S+S^T)\neq 0,
$$
the restriction of the  Korotkin--Samtleben bracket to the entries of
the Stokes matrix leads to a closed Poisson algebra given by the
formulae (\ref{poissonu}).

The Casimirs of this Poisson bracket are the eigenvalues of the
matrix $S^{-T}S$ so that the generic Poisson leaves $\mathcal
L_{generic}$ have dimension
$$
\dim(\mathcal L_{generic})=\frac{n(n-1)}{2}-\left[\frac{n}{2}\right].
$$

\begin{remark}\label{rm:nu}
Note that actually it is not necessary to choose $q=1$. In fact,
given any $q$ such that condition (\ref{eq:det-condq}) is satisfied,
it is always true that
\bea
\left\{s_{ik},s_{jl}\right\}&=&\frac{1}{q^2 s_{ik}s_{jl}}\{\Tr(M_i M_k),\Tr(M_j M_l)\}=\nn\\
&=&
\frac{i\pi\left(\epsilon(l-k)+\epsilon(k-j)+\epsilon(i-l)+\epsilon(j-i)\right)}{q^2 s_{ik}s_{jl}}
\tr\left([M_k, M_i][M_j,M_l]\right),\nn
\eea
where $\epsilon(k)$ is the sign of $k$.
By brute force computation, using the specific form of the
matrices $M_k$, one obtains always the same Poisson bracket
(\ref{poissonu}). This observation is  quite important when we want
to study the case when the rank of the matrix $S+S^T$ is very low.
In this case we can pick $q\neq 1$ and prove that the Poisson
algebra is (\ref{poissonu}) anyway. We will discuss this case
further in Section \ref{se:bondal}.
\end{remark}

\section{Poisson algebras of geodesic length functions}\label{se:geo}

In this Section we discuss the Dubrovin--Ugaglia bracket in the
context of Teichm\"uller theory. We first recall some key facts which
will be needed below.

Due to E. Verlinde and H. Verlinde \cite{VV} the
configuration space of Einstein gravity in $2+1$ dimensions is a  Riemann surface with boundary
components (or holes) and orbifold points times an interval representing
the time variable. The algebra of observables is identified with the collection of geodesic length
functions of geodesic representatives of homotopy classes of
closed curves together with its natural mapping class
group action.

The Poisson structure on geodesic length functions is provided
by the Goldman brackets \cite{Gold} and coincides with
the Poisson brackets that follow from the Chern--Simons theory \cite{FR}.

The Poisson algebra of geodesic functions is always closed
(and linear) on the subset of geodesic functions corresponding to
{\it multi-curves,}\/ which are sets of curves without intersections and
self-intersections. However, these sets are always
infinite whereas the Teichm\"uller spaces $\mathcal T_{g,s,n}$ are spaces of
(real) dimension $6g-6+2s+2n$, where $g$ is the genus of the Riemann
surface, $s$ is the number of boundary components (or holes) and $n$ is the number of orbifold points.

Multi-curve geodesic functions are therefore algebraically dependent,
and one encounters the problem of constructing an algebraically
independent (or, at least, finite) basis of observables such that
the Poisson brackets become closed on this set. In the general case this problem
is still open.

In the special case of Riemann surfaces with one or two holes
\cite{ChF2},~\cite{ChP}, and in the case of a Riemann sphere with one hole and
$n$ orbifold points of order $2$ \cite{Ch1a}, the Poisson algebra
generated by the Goldman bracket on  geodesic length functions
closes and  coincides  with the Dubrovin--Ugaglia bracket.

Here we recall the basics of this construction, which is based on the graph description of the Teichm\"uller space. Denote by
$\Sigma_{g,s,n}$ a Riemann surface of genus $g$ with $s$ holes and
$n$ orbifold points of order two.
We assume the hyperbolicity condition $2g-2+s>0$, so that by the Poincar\'e uniformization
theorem, we have
$$
\Sigma_{g,s,n}\sim {\mathbb H}^+_2\slash \Delta_{g,s,n},
$$
where $ {\mathbb H}^+_2$ is the upper half plane and $\Delta_{g,s,n}$ is a Fuchsian group,  the fundamental group of the surface~$\Sigma_{g,s,n}$:
$$
 \Delta_{g,s,n}=\langle\gamma_1\dots,\gamma_{2g+s+n-1}\rangle,\qquad
 \gamma_1\dots,\gamma_{2g+s+n-1}\in PSL(2,{\mathbb R}).
$$
In particular, for orbifold Riemann
surfaces, the Fuchsian group $\Delta_{g,s,n}$ is such that all its elements are either hyperbolic or have
trace equal to zero.

We recall the Thurston shear-coordinate description \cite{Penn1},~\cite{Fock1}  of the
Teichm\"uller spaces of Riemann surfaces with holes and, possibly, orbifold points (see \cite{Ch1a}).
The main idea is to decompose each hyperbolic matrix
$\gamma\in \Delta_{g,s,n}$ as a product of the form
\begin{equation}\label{eq:decomp}
\gamma= (-1)^K R^{k_{i_p}} X_{Z_{i_p}} \dots R^{k_{i_1}} X_{Z_{i_1}},\qquad i_j\in I,\quad
k_{i_j}=1,2,\quad K:=\sum_{j=1}^p k_{i_j}
\end{equation}
where $I$ is a set of integer indices and the matrices $R,\, L$ and $X_{Z_i}$ are defined as follows:
\begin{eqnarray}\nn\label{eq:generators}
&&
R:=\left(\begin{array}{cc}1&1\\-1&0\\
\end{array}\right), \qquad
L=-R^2:=\left(\begin{array}{cc}0&1\\-1&-1\\
\end{array}\right), \\
&&
X_{Z_i}:=\left(\begin{array}{cc}0&-\exp\left({\frac{Z_i}{2}}\right)\\
\exp\left(-{\frac{Z_i}{2}}\right)&0\\
\nn\end{array}\right),
\end{eqnarray}
and to decompose each traceless element as
\begin{equation}
\label{eq:decomp1}
\gamma_0=\gamma^{-1} F \gamma,
\end{equation}
where $\gamma$ is decomposed as in (\ref{eq:decomp}) and
$$
F=\left(\begin{array}{cc}0&1\\-1&0\\
\end{array}\right).
$$
The main point of this construction is that one can obtain the decompositions (\ref{eq:decomp}) and
(\ref{eq:decomp1}) by looking at closed loops on the fat--graph. The fat--graph, or spine, $\Gamma_{g,s,n}$ is a connected graph that
can be drawn without self-intersections on $\Sigma_{g,s,n}$ that has all
vertices of valence three except exactly $n$ one-valent vertices situated at the
orbifold points, has a prescribed cyclic ordering of
labeled edges entering each vertex, and it is a maximal graph in the
sense that its complement on the Riemann surface
is a set of disjoint polygons (faces), each polygon
containing exactly one hole (and becoming simply connected after gluing
this hole). Since a graph must have at least one face, only Riemann
surfaces with holes, $s>0$, can be described in this way.
These fat graphs (or spines)
constructed originally in~\cite{Fock1}~\cite{Fock2} in the case of surfaces
without orbifold points are dual to ideal triangle decompositions of Penner \cite{Penn1}.

We obtain the decomposition of an element of the Fuchsian group
$\Delta_{g,s,n}$ using the one-to-one correspondence between closed paths in the
fat graph (spine) $\Gamma_{g,s,n}$ and conjugacy classes of the Fuchsian group $\Delta_{g,s,n}$.
The decomposition (\ref{eq:decomp}) can
be obtained by establishing a one-to-one correspondence between
elements of the Fuchsian group itself and closed paths in the spine
starting and terminating at the same directed edge. Each  time the
path $A$ corresponding to the element $\gamma_A$ (or, equivalently,
to its invariant closed geodesic)  passes through the
$\alpha$th edge, an edge-matrix  $X_{Z_\alpha}$ with the real
coordinate $Z_\alpha$ (related to the length of that edge)
 appears in the decomposition of $\gamma$. At
the end of the edge, the path can either turn right or left, and a
matrix $R$ or $L$ respectively appears in the decomposition
\cite{Fock1}. To obtain decomposition (\ref{eq:decomp1}), we observe that
when a path reaches a one-valent vertex (a
pending vertex), it undergoes an {\em
inversion} \cite{Ch1}, which corresponds to inserting the matrix $F$
into the corresponding string of $2\times2$-matrices.
The edge terminating at a pending vertex is called a {\em pending}
edge.

The  algebras of geodesic length functions were constructed in \cite{Ch1}
by  postulating the Poisson relations on the level of the shear
coordinates $Z_\alpha$ of the Teichm\"uller space:
\begin{equation}
\label{eq:Poisson}
\bigl\{f({\mathbf Z}),g({\mathbf Z})\bigr\}=\sum_{{\hbox{\small 3-valent} \atop \hbox{\small vertices $\alpha=1$} }}^{4g+2s+n-4}
\,\sum_{i=1}^{3\!\!\mod 3}
\left(\frac{\partial f}{\partial Z_{\alpha_i}} \frac{\partial g}{\partial Z_{\alpha_{i+1}}}
- \frac{\partial g}{\partial Z_{\alpha_i}} \frac{\partial f}{\partial Z_{\alpha_{i+1}}}\right),
\end{equation}
where the sum ranges all the {\em three-valent} vertices of a graph and
$\alpha_i$ are the labels of the cyclically (counterclockwise)
ordered ($\alpha_{i+3}\equiv \alpha_i $) edges incident to the vertex
with the label $\alpha$. This bracket gives rise to the {\it Goldman
bracket} on the space of geodesic length functions \cite{Gold}.

In terms of geodesic length functions the bracket (\ref{eq:Poisson}) corresponds to
\begin{equation}\label{eq:poisson1}
\left\{ \Tr \gamma_A, \Tr  \gamma_B\right\}=\frac{1}{2}\Tr ( \gamma_A  \gamma_B ) -\frac{1}{2}  \Tr ( \gamma_A  \gamma_B^{-1}).
\end{equation}
So we see that every time we consider the bracket between the
geodesics lengths of two loops $A$ and $B$, we produce the geodesics
lengths of two new loops $A\,B$ and $A\,B^{-1}$. To close the Poisson algebra one must use the
{\em skein relation} valid for two arbitrary matrices in $PSL(2)$:
\begin{equation}\label{eq:skein}
 \Tr \gamma_A \Tr  \gamma_B=\Tr ( \gamma_A  \gamma_B ) + \Tr ( \gamma_A  \gamma_B^{-1}) .
\end{equation}
We can use this relation for resolving the crossing between the two geodesics $A$ and $B$ as in Fig. \ref{fi:skein-cl}.

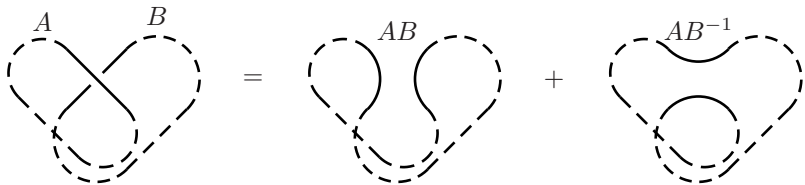
\begin{figure}[h]
\hspace*{2cm}
\vskip .2in
{\psset{unit=0.4}
\begin{pspicture}(-12,-3)(12,2)
\newcommand{\LOOPS}{%
\psarc[linestyle=dashed, linewidth=1pt](2,0){1.42}{-45}{135}
\pcline[linestyle=dashed, linewidth=1pt](3,-1)(1,-3)
\psarc[linestyle=dashed, linewidth=1pt](0,-2){1.42}{135}{315}
\psarc[linestyle=dashed, linewidth=1pt](0.2,-1.8){1.13}{-135}{45}
\pcline[linestyle=dashed, linewidth=1pt](-0.6,-2.6)(-2.6,-0.6)
\psarc[linestyle=dashed, linewidth=1pt](-1.8,0.2){1.13}{45}{225}%%geodesic
}
\rput(-10,0){\LOOPS}
\rput(-10,0){
\pcline[linewidth=1pt](-1,1)(1,-1)
\pcline[linewidth=1pt](-1,-1)(-0.2,-0.2)
\pcline[linewidth=1pt](1,1)(0.2,0.2)
\rput(-1.8,1.6){\makebox(0,0)[cb]{$A$}}
\rput(2,1.8){\makebox(0,0)[cb]{$B$}}
}
\rput(-4.8,0){\makebox(0,0)[cc]{$=$}}
\rput(0,0){\LOOPS}
\rput(0,0){
\psarc[linewidth=1pt](-2,0){1.42}{-45}{45}
\psarc[linewidth=1pt](2,0){1.42}{135}{225}
\rput(0,1.3){\makebox(0,0)[cb]{$AB$}}
}
\rput(5.2,0){\makebox(0,0)[cc]{$+$}}
\rput(10,0){\LOOPS}
\rput(10,0){
\psarc[linewidth=1pt](0,2){1.42}{225}{315}
\psarc[linewidth=1pt](0,-2){1.42}{45}{135}
\rput(0,1.3){\makebox(0,0)[cb]{$A B^{-1}$}}
}
\end{pspicture}
}
\caption{\small The classical skein relation.}
\label{fi:skein-cl}
\end{figure}

The skein relation is often not enough to close the Poisson
algebra on a finite set of generators. In this paper, we shall
consider two special cases in which we indeed can close the
algebra just by means of skein relation: the case of Riemann
surfaces of genus $g$ and one or two holes (which we dub $CFP$ due
to the fact that it was mostly developed in \cite{ChP,ChF}) in
subsection \ref{se:CP} and  the case of a Riemann sphere with one
hole and with $n\geq3$ orbifold points of order two (dubbed $\mathcal A_n$
case due to its close ties to cluster algebra theory \cite{FST})  in subsection
\ref{sub:An}.

\subsection{The $CFP$ case}\label{se:CP}

This is  the case of a Riemann surface of genus $g$ with one or two holes, the  fat-graph
on which graph-simple geodesics constitute a convenient algebraic basis is shown in Fig.~\ref{octopus}.
The genus $g=\left[\frac{n-1}{2}\right]$, where $n$ is the number of vertical edges, and the number
$s$ of holes is
$$
s=\left\{\begin{array}{lc}
1&\hbox{for $n$ odd,}\\
2&\hbox{for $n$ even.}\\
\end{array}\right.
$$
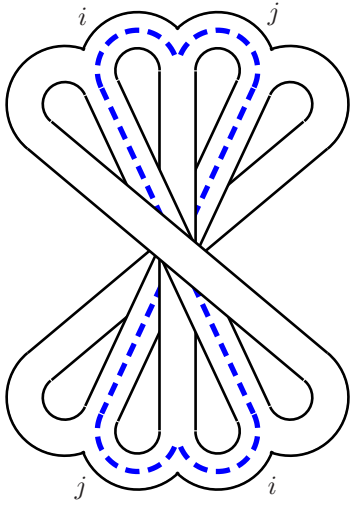
\begin{figure}
{\psset{unit=0.6}
\begin{pspicture}(-5.5,-5.5)(5.5,5.5)
\newcommand{\LINE}{%
\psframe[linecolor=white, fillstyle=solid, fillcolor=white](-.4,4)(.4,-4)
\pcline[linewidth=1pt](-.4,4)(-.4,-4)
\pcline[linewidth=1pt](.4,4)(.4,-4)
}
\newcommand{\LINELINE}{%
\psframe[linecolor=white, fillstyle=solid, fillcolor=white](-.4,4)(.4,-4)
\pcline[linewidth=1pt](-.4,4)(-.4,-4)
\pcline[linewidth=1pt](.4,4)(.4,-4)
\pcline[linewidth=2pt,linecolor=blue,linestyle=dashed](0,4)(0,-4)
}
\newcommand{\CIRCLE}{%
\psarc[linewidth=1pt](0,4.1){1.29}{-12.5}{192.5}
}
\rput{-50}(0,0){\LINE}
\rput{-25}(0,0){\LINELINE}
\rput(0,0){\LINE}
\rput{25}(0,0){\LINELINE}
\rput{50}(0,0){\LINE}
\newcommand{\PAT}{%
\rput{-37.5}(0,0){\psarc[linewidth=1pt](0,4.1){1.29}{-12.5}{192.5}}
\rput{-12.5}(0,0){\psarc[linewidth=1pt](0,4.1){1.29}{-12.5}{192.5}}
\rput{12.5}(0,0){\psarc[linewidth=1pt](0,4.1){1.29}{-12.5}{192.5}}
\rput{37.5}(0,0){\psarc[linewidth=1pt](0,4.1){1.29}{-12.5}{192.5}}
\rput{-37.5}(0,0){\pswedge[linecolor=white, fillstyle=solid, fillcolor=white](0,4.1){1.25}{-12.5}{192.5}}
\rput{-12.5}(0,0){\pswedge[linecolor=white, fillstyle=solid, fillcolor=white](0,4.1){1.25}{-12.5}{192.5}}
\rput{12.5}(0,0){\pswedge[linecolor=white, fillstyle=solid, fillcolor=white](0,4.1){1.25}{-12.5}{192.5}}
\rput{37.5}(0,0){\pswedge[linecolor=white, fillstyle=solid, fillcolor=white](0,4.1){1.25}{-12.5}{192.5}}
\rput{-37.5}(0,0){\psarc[linewidth=1pt](0,4.1){.49}{-12.5}{192.5}}
\rput{-12.5}(0,0){\psarc[linewidth=1pt](0,4.1){.49}{-12.5}{192.5}}
\rput{12.5}(0,0){\psarc[linewidth=1pt](0,4.1){.49}{-12.5}{192.5}}
\rput{37.5}(0,0){\psarc[linewidth=1pt](0,4.1){.49}{-12.5}{192.5}}
\rput{12.5}(0,0){\psarc[linewidth=2pt,linecolor=blue,linestyle=dashed](0,4.1){.89}{10}{192.5}}
\rput{-12.5}(0,0){\psarc[linewidth=2pt,linecolor=blue,linestyle=dashed](0,4.1){.89}{-12.5}{170}}
}
\rput(0,0){\PAT}
\rput{180}(0,0){\PAT}
\rput(-2,5){\makebox(0,0)[rb]{$i$}}
\rput(2,5){\makebox(0,0)[lb]{$j$}}
\rput(-2,-5){\makebox(0,0)[rt]{$j$}}
\rput(2,-5){\makebox(0,0)[lt]{$i$}}
\end{pspicture}
}
\caption{\small{The fat graph in the $CFP$ case, the blue geodesic is $G{i,j}$.}}
\label{octopus}
\end{figure}
Graph-simple closed geodesics in this picture are those and only
those that pass through exactly two different vertical edges; we can
then enumerate them by ordered pairs of edge indices denoting by
$G_{ij}$ ($i<j$) the corresponding geodesic functions. Denoting by
$Z_1,\dots,Z_n$ the coordinates on the vertical edges and by
$Y_1,\dots,Y_{2n-6}$ those on the horizontal edges, we obtain
\be
\label{eq:geoNR}
G_{ij}=    X_{Z_i} L Y_{n+i-4} \dots
R X_{Y_{n+j-5}} L  X_{Z_j} R X_{Y_{j-2}}\dots X_{Y_{i}} L X_{Y_{i-1}} R,
\ee
so, for example,
\begin{eqnarray}
&{}&
G_{12} =X_{Z_1} L X_{Z_2} R,\nn\\
&{}&
G_{13} = X_{Z_1} R X_{Y_{n-2}} L X_{Z_3} R X_{Y_{1}} L,\nn\\
&{}&
\dots\nn\\
&{}&
G_{1n}= X_{Z_1} R X_{Y_{n-2}} R X_{Y_{n-1}}\dots R X_{Y_{2n-6}} L X_{Z_n}  L X_{Y_{n-3}}  L \dots  X_{Y_{1}}  L,\nn\\
&{}&
G_{23} =   X_{Z_2} L X_{Y_{n-2}} L  X_{Z_3} R X_{Y_{1}} R.\nn
\end{eqnarray}
The Poisson algebra
for the functions $G_{ij}$ is described by
\be
\{G_{ij},G_{kl}\}=\left\{\begin{array}{l}
0,\quad j<k,\\
0,\quad k<i,\ j<l,\\
G_{ik}G_{jl}-G_{kj}G_{il},\quad i<k<j<l,\\
\frac12 G_{ij}G_{jl}-G_{il},\quad j=k,\\
G_{il}-\frac12 G_{ij}G_{il},\quad i=k,\ j<l\\
G_{ik}-\frac12 G_{ij}G_{kj},\quad j=l,\ i<k.
\end{array}
\right. \label{P-geod}
\ee
This is just a rescaled Dubrovin--Ugaglia bracket .

\subsection{The $\mathcal A_n$ case}\label{sub:An}

The simplest case of orbifold Riemann surface is a Riemann sphere
$\Sigma_{0,1,n}$ with one hole and $n\geq3$ orbifold points of
order two. In this case, the fat-graph $\Gamma_{0,1,n}$ is a
tree-like graph with $n$ pending vertices depicted in
Fig.~\ref{fi:An} for $n=3,4$. We enumerate the $n$
pending vertices counterclockwise, $i,j=1,\dots, n$, and consider
the algebra of all geodesic functions.

\begin{figure}[h]
\begin{center}
{\psset{unit=0.6}
\begin{pspicture}(-4,-3)(8,3)
%A3
\newcommand{\PAT}{%
\pcline(-0.87,0.5)(-0.87,2)
\pcline(0.87,0.5)(0.87,2)
\pscircle*(0,2){0.1}
}
\newcommand{\GEOD}{%
\psarc[linewidth=1.5pt,linestyle=dashed](0,2){0.6}{0}{180}
\psarc[linewidth=1.5pt,linestyle=dashed](1.73,-1){0.4}{-120}{60}
\psbezier[linewidth=1.5pt,linestyle=dashed](0.6,2)(0.4,.2)(.4,.2)(1.93,-0.65)
\psbezier[linewidth=1.5pt,linestyle=dashed](-0.6,2)(-0.6,0.5)(.53,-.8)(1.53,-1.35)
}
\rput(-6,0){\makebox(0,0){$\mathbf{n{=}3}$}}
\rput(-3,0){\PAT}
\rput{120}(-3,0){\PAT}
\rput{240}(-3,0){\PAT}
\rput(-3,0){
\psarc[linewidth=1.5pt,linestyle=dashed,linecolor=red](0,2){0.6}{0}{180}
\psarc[linewidth=1.5pt,linestyle=dashed,linecolor=red](1.73,-1){0.4}{-120}{60}
\psbezier[linewidth=1.5pt,linestyle=dashed,linecolor=red](0.6,2)(0.4,.2)(.4,.2)(1.93,-0.65)
\psbezier[linewidth=1.5pt,linestyle=dashed,linecolor=red](-0.6,2)(-0.6,0.5)(.53,-.8)(1.53,-1.35)
}
\rput{120}(-3,0){
\psarc[linewidth=1.5pt,linestyle=dashed,linecolor=blue](0,2){0.6}{0}{180}
\psarc[linewidth=1.5pt,linestyle=dashed,linecolor=blue](1.73,-1){0.4}{-120}{60}
\psbezier[linewidth=1.5pt,linestyle=dashed,linecolor=blue](0.6,2)(0.4,.2)(.4,.2)(1.93,-0.65)
\psbezier[linewidth=1.5pt,linestyle=dashed,linecolor=blue](-0.6,2)(-0.6,0.5)(.53,-.8)(1.53,-1.35)
}
\rput{240}(-3,0){
\psarc[linewidth=1.5pt,linestyle=dashed,linecolor=green](0,2){0.6}{0}{180}
\psarc[linewidth=1.5pt,linestyle=dashed,linecolor=green](1.73,-1){0.4}{-120}{60}
\psbezier[linewidth=1.5pt,linestyle=dashed,linecolor=green](0.6,2)(0.4,.2)(.4,.2)(1.93,-0.65)
\psbezier[linewidth=1.5pt,linestyle=dashed,linecolor=green](-0.6,2)(-0.6,0.5)(.53,-.8)(1.53,-1.35)
}
\rput(-3,1.6){\makebox(0,0){\tiny$\mathbf1$}}
\rput(-4.4,-.8){\makebox(0,0){\tiny$\mathbf2$}}
\rput(-1.6,-.8){\makebox(0,0){\tiny$\mathbf3$}}
\pcline[linewidth=0.5pt]{->}(-4.9,-2.1)(-4.7,-1.6)
\rput(-5,-2.5){\makebox(0,0){\small$G_{1,2}$}}
\pcline[linewidth=0.5pt]{->}(-0.6,-2)(-0.8,-1.4)
\rput(-.5,-2.5){\makebox(0,0){\small$G_{2,3}$}}
\pcline[linewidth=0.5pt]{->}(-2,2.5)(-2.5,2.3)
\rput(-1.5,2.6){\makebox(0,0){\small$G_{1,3}$}}
%A4
\rput(2,0){\makebox(0,0){$\mathbf{n{=}4}$}}
\rput{-90}(5,0){\PAT}
\rput{-210}(5,0){\PAT}
\rput{-330}(5,0){\PAT}
\rput{90}(9,0){\PAT}
\rput{210}(9,0){\PAT}
\rput{330}(9,0){\PAT}
\pscircle*[linecolor=white, fillstyle=solid, fillcolor=white](7,0){0.3}
\rput{30}(5,0){
\psarc[linewidth=1.5pt,linestyle=dashed,linecolor=red](0,2){0.6}{0}{180}
%\psarc[linewidth=1.5pt,linestyle=dashed,linecolor=red](1.73,-1){0.4}{-120}{60}
\psbezier[linewidth=1.5pt,linestyle=dashed,linecolor=red](0.6,2)(0.4,.2)(.4,.2)(1.93,-0.65)
\psbezier[linewidth=1.5pt,linestyle=dashed,linecolor=red](-0.6,2)(-0.6,0.5)(.53,-.8)(1.53,-1.35)
}
\rput{-150}(9,0){
\psarc[linewidth=1.5pt,linestyle=dashed,linecolor=red](0,2){0.6}{0}{180}
%\psarc[linewidth=1.5pt,linestyle=dashed,linecolor=red](1.73,-1){0.4}{-120}{60}
\psbezier[linewidth=1.5pt,linestyle=dashed,linecolor=red](0.6,2)(0.4,.2)(.4,.2)(1.93,-0.65)
\psbezier[linewidth=1.5pt,linestyle=dashed,linecolor=red](-0.6,2)(-0.6,0.5)(.53,-.8)(1.53,-1.35)
}
\rput{-30}(9,0){
\psarc[linewidth=1.5pt,linestyle=dashed,linecolor=blue](0,2){0.6}{0}{180}
%\psarc[linewidth=1.5pt,linestyle=dashed,linecolor=red](1.73,-1){0.4}{-120}{60}
\psbezier[linewidth=1.5pt,linestyle=dashed,linecolor=blue](-0.6,2)(-0.4,.2)(-.4,.2)(-1.93,-0.65)
\psbezier[linewidth=1.5pt,linestyle=dashed,linecolor=blue](0.6,2)(0.6,0.5)(-.53,-.8)(-1.53,-1.35)
}
\rput{-210}(5,0){
\psarc[linewidth=1.5pt,linestyle=dashed,linecolor=blue](0,2){0.6}{0}{180}
%\psarc[linewidth=1.5pt,linestyle=dashed,linecolor=red](1.73,-1){0.4}{-120}{60}
\psbezier[linewidth=1.5pt,linestyle=dashed,linecolor=blue](-0.6,2)(-0.4,.2)(-.4,.2)(-1.93,-0.65)
\psbezier[linewidth=1.5pt,linestyle=dashed,linecolor=blue](0.6,2)(0.6,0.5)(-.53,-.8)(-1.53,-1.35)
}
\rput(4.2,1.3){\makebox(0,0){\tiny$\mathbf1$}}
\rput(4.2,-1.3){\makebox(0,0){\tiny$\mathbf2$}}
\rput(9.8,-1.3){\makebox(0,0){\tiny$\mathbf3$}}
\rput(9.8,1.3){\makebox(0,0){\tiny$\mathbf4$}}
%\pcline[linewidth=0.5pt]{->}(3.6,-1)(3.8,-0.2)
\rput(2.5,-1.8){\makebox(0,0){\small$G_{2,4}$}}
%\pcline[linewidth=0.5pt]{->}(9.4,-1)(9.2,-0.4)
\rput(11.5,-1.8){\makebox(0,0){\small$G_{1,3}$}}
\end{pspicture}
}
\caption{\small Generating graphs for $\mathcal A_n$ algebras for $n=3,4$. We indicate character
geodesics whose geodesic functions $G_{ij}$ enter bases of the corresponding algebras.}
\label{fi:An}\end{center}
\end{figure}
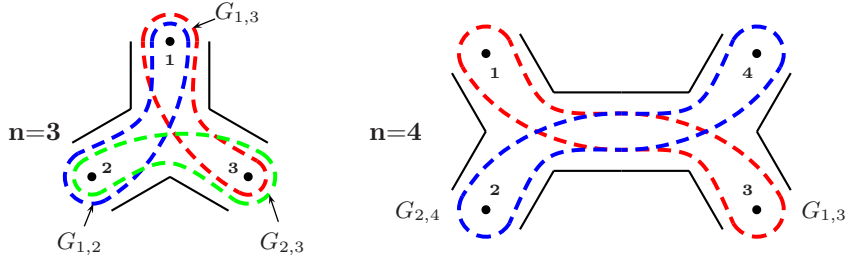

We consider a basis $\gamma_1,\dots,\gamma_n$ in the Fuchsian group $\Delta_{0,1,n}$ such that
\be\label{eq:geoAn}
-\Tr(\gamma_i\gamma_j)=G_{i,j}.
\ee
(The sign convention is such that when we interpret $G_{i,j}$ as being the geodesic functions
related to lengths $\ell_{i,j}$ of closed geodesics, we have $G_{i,j}=2\cosh(\ell_{i,j}/2)\ge2$.)  In this case, for convenience
we let $Z_i$ denote the coordinates of pending edges and $Y_j$ all
other coordinates. This basis in the Fuchsian group $\Delta_{0,1,n}$ is given by the following:
\begin{eqnarray}
&{}&
\gamma_1=F,\nn\\
&{}&
\gamma_2 =- X_{Z_1} L X_{Z_2} F X_{Z_2} R X_{Z_1}\nn\\
&&
\gamma_3= -X_{Z_1} R X_{Y_1} L X_{Z_3} F X_{Z_3} R X_{Y_1} L X_{Z_1}   \nn\\
&{}&
\dots\nn \\
&{}&
\gamma_i =- X_{Z_1} R X_{Y_1} R X_{Y_2}  \dots R X_{Y_{i-2}} L X_{Z_i} F X_{Z_i}  R X_{Y_{i-2}} L  \dots
X_{Y_1} L X_{Z_1},\label{eq:basisn}\\
&{}&
\dots \nn \\
&{}&
\gamma_{n-1} =- X_{Z_1} R X_{Y_1} R X_{Y_2}  \dots R X_{Y_{n-3}} L X_{Z_{n-1}} F X_{Z_{n-1}}  R X_{Y_{n-3}} L  \dots
X_{Y_1} L X_{Z_1},\nn \\
&{}&
\gamma_n =-  X_{Z_1} R X_{Y_1} R X_{Y_2}  \dots R X_{Y_{n-3}} R X_{Z_n} F X_{Z_n}  R X_{Y_{n-3}} L  \dots X_{Y_1} L X_{Z_1},\nn
\end{eqnarray}
Observe that $\Tr\gamma_i=0$, $i=1,\dots,n$. It is not hard to check that the matrix
$$
\gamma_\infty:=(\gamma_1 \gamma_{2} \dots \gamma_n)^{-1}
$$
has eigenvalues $(-1)^{n-1} e^{\pm P/2}$, where $P$ is the length of the perimeter around the hole:
\begin{equation}\label{eq:per}
P = 2\sum_{i=1}^n Z_i+2\sum_{j=1}^{n-3} Y_j.
\end{equation}

Let $G_{i,j}=-\Tr(\gamma_i \gamma_j)$ with $i<j$
denote the geodesic function corresponding to the geodesic line that
encircles exactly two pending vertices with the indices $i$ and
$j$. Examples for $n=3$ and $n=4$ are in figure \ref{fi:An}.
It turns out that these geodesic functions suffice for closing the
Poisson algebra:
\begin{eqnarray}\label{eq:NR}
&&
\left\{G_{i,k},G_{j,l}\right\}=0,\quad\hbox{for}\,\,i<k<j<l,\quad\hbox{and for}\,\, i<j<l<k, \nn\\
&&
\left\{G_{i,k},G_{j,l}\right\}=2\left(G_{i,j}G_{k,l}-G_{i,l}G_{k,j}\right),\quad\hbox{for}\, \, i<j<k<l,\nn\\
&&
\left\{G_{i,k},G_{k,l}\right\}=G_{i,k}G_{k,l}-2G_{i,l},\quad\hbox{for}\,\,  i<k<l,\\
&&
\left\{G_{i,k},G_{j,k}\right\}=-\left(G_{i,k}G_{j,k}-2G_{i,j}\right),\quad\hbox{for}\, \, i< j<k,\nn\\
&&
\left\{G_{i,k},G_{i,l}\right\}=-\left(G_{i,k}G_{i,l}-2G_{k,l}\right),\quad\hbox{for}\, \, i<k<l.\nn
\end{eqnarray}
Note that this is again a simple rescaling of the Dubrovin--Ugaglia bracket.

\begin{remark}\label{rk:g}
The formulae for $G_{ij}$ in terms of the shear
coordinates $Z_1,\dots,Z_n$, $Y_1,\dots,Y_{n-3}$ in the $\mathcal A_n$ case
coincide with a specialization of the formulae of the geodesics
$G_{ij}$ given by (\ref{eq:geoNR}) in which we assume $Y_{n-3+i}=Y_i$
for $ i=1,\dots n-3$ and we take  the double lengths
$2{Z_1},\dots,2Z_n$ In other words:
$$
G^{(\mathcal A_n)}_{ij}(Z_1,\dots,Y_{n-3})=G^{CFP}_{ij}(2 Z_1,\dots 2Z_n,Y_1,\dots,Y_{n-3},Y_1,\dots,Y_{n-3}).
$$
\end{remark}

\section{Symplectic leaves corresponding to the Teichm\"uller space}\label{se:bondal}

As mentioned in the introduction, since the Poisson bracket for
geodesic length functions both  in the $CFP$ case (Riemann surface
$\Sigma_{g,s,0}$ of genus $g$ with $s=1,2$ holes) and in the $\mathcal A_n$
case (Riemann sphere $\Sigma_{0,1,n}$ with one hole and $n$ orbifold
points of order two) coincides with the Dubrovin--Ugalgia bracket,
it makes sense to characterize the symplectic leaves to which the
two Teichm\"uller spaces $\mathcal T_{g,s,0}$, $s=1,2$, and
$\mathcal T_{0,1,n}$ belong.

In particular we recall that
$$
\dim_{\mathbb R}\left(\mathcal T_{g,s,0}
\right)= \left\{\begin{array}{ll}
3n-7 &\hbox{for $n$ odd}\\
3n-8 &\hbox{for $n$ even,}\\
\end{array}\right.\,\hbox{where}\,\, g=\left[\frac{n-1}{2}\right] , \,
s=\left\{\begin{array}{ll}
1&\hbox{for $n$ odd,}\\
2&\hbox{for $n$ even,}\\
\end{array}\right.
$$
and
$$
\dim_{\mathbb R}\left(\mathcal T_{0,1,n}\right)= 2(n-2),
$$
while the generic symplectic leaves $\mathcal L_{\small{generic}}$ in the Dubrovin--Ugaglia bracket have dimension
$$
\dim_{\mathbb C}\left(\mathcal L_{{generic}}\right)=\frac{n(n-1)}{2}-\left[\frac{n}{2}\right].
$$
It is natural to ask whether the Teichm\"uller spaces arise as real
slices of some  sub-varieties of a generic leaf or of a degenerated
leaf. In this section we prove that in both cases the
Teichm\"uller spaces correspond to degenerated symplectic leaves
complex  dimension equal to the real dimension of the Teichm\"uller
space itself:

\begin{theorem}\label{th:main0}
Denote by $\mathcal L_{\mathcal A_n}$ and by $\mathcal L_{CFP}$  the
symplectic leaves to which the Stokes matrices of with entries
$s_{ij}=G_{ij}$ where $G_{ij}$ are given respectively
 by (\ref{eq:geoAn},\ref{eq:basisn}) or by (\ref{eq:geoNR}) belong. Then
$$
\dim_{\mathbb C}\left(\mathcal L_{CFP}
\right)= \left\{\begin{array}{lc}
3n-7 &\hbox{ for $n$ odd}\\
3n-8 &\hbox{ for $n$ even}\\
\end{array}\right.
$$
and
$$
\dim_{\mathbb C}\left(\mathcal L_{\mathcal A_n}\right)= 2(n-2),
$$
\end{theorem}

\proof In order to compute the dimension of the symplectic leaf to which a
particular Stokes matrix belongs we use a
formula by Bondal \cite{Bondal} which is based on the block diagonal
form of the Jordan normal form $J_0$ of $S^{-T}S$:

\begin{lm}
Given an arbitrary upper triangular matrix $S$ with $1$ on the diagonal, the Jordan normal form $J_0$ of $S^{-T}S$ decomposes as follows
\begin{equation}\label{eq:j-dec}
J_0=\sum_{\lambda\neq(-1)^{k+1}} n_{\lambda,k} \left(J_{\lambda,k}
\oplus J_{\frac{1}{\lambda}, k} \right) +
\sum_{\lambda=(-1)^{k+1}} m_{\lambda,k} J_{\lambda,k},
\end{equation}
where $J_{\lambda,k}$ denotes the $k\times k$ Jordan block with eigenvalue $\lambda$, i.e.
$$
J_{\lambda,k}=\left(\begin{array}{ccccc}
\lambda&1&0&\dots&0\\
0&\lambda&1&\dots&0\\
\dots&\dots&\dots&\dots&\dots\\
0&\dots&0&\lambda&1\\
0&\dots&0&0&\lambda\\
\end{array}\right),
$$
and $n_{\lambda,k}$ and $m_{(-1)^{k+1},k}$ are the multiplicities of the blocks
$J_{\lambda,k} \oplus J_{\frac{1}{\lambda}, k}$ and $J_{(-1)^{k+1},k}$ respectively.

The dimension of the symplectic leaf $\mathcal L_S$ to which $S$ belongs is
$$
\dim_{\mathbb C}\left(\mathcal L_S\right)=\frac{n(n-1)}{2}-{\rm d}(S),
$$
where
\begin{eqnarray}\label{eq:bondal}\nn
&&
{\rm d}(S)= \sum_{\lambda\neq\pm1} \min(k,l) n_{\lambda,k}n_{\lambda,l}+
2  \sum_{\lambda=\pm1} \min(k,l) n_{\lambda,k}n_{\lambda,l}+\nn \\
&&
\quad\qquad+2 \sum_{\lambda=\pm1} \min(k,l) n_{\lambda,k}m_{\lambda,l}+
 \frac{1}{2} \sum_{\lambda=\pm1} \min(k,l) m_{\lambda,k}m_{\lambda,l}-\\
 &&
\quad\qquad - \frac{1}{2} \sum_{\lambda=1} m_{1,l}+
 \sum_{\lambda=\pm1} k\, n_{\lambda,k}\nn
 \end{eqnarray}
\end{lm}

\begin{proof}
The proof of the first statement is a trivial consequence of Section
5.5 in Bondal's paper. The formula (\ref{eq:bondal}) is  (5.10) in
\cite{Bondal} (with two small corrections: a factor $2$ in the first
term of the second row and the last term in the last row were
missing).
\end{proof}

In order to use this result to compute the dimension of our
symplectic leaves we need to describe the Jordan normal form $J_0$
of $S^{-T}S$ for a Stokes matrix $S$ with entries $s_{ij}=G_{ij}$
where $G_{ij}$ are given either by (\ref{eq:geoNR}) or by
(\ref{eq:geoAn},\ref{eq:basisn}). This is achieved in the next two
theorems which will be proved in subsections \ref{suse:proof1} and
\ref{suse:proof2} respectively.

\begin{theorem}\label{thm:JordanAn}
Let $S$ be un upper triangular matrix with $1$ on the diagonal and off diagonal entries
$$
S_{ij}= -\Tr(\gamma_i \gamma_j),\qquad i<j,
$$
with $\gamma_1,\dots,\gamma_n$ given in terms of shear coordinates
by formula (\ref{eq:basisn}). Then  for $n$ even, the matrix
$S^{-T}S$, has the following Jordan form:
\begin{equation}\label{eq:J-Veven}
J_0=\left(\begin{array}{ccc}
\begin{array}{cc}
-e^{P}&0\\
0&-e^{-P}\\
\end{array}& {\mathbb O}&{\mathbb O}\\
{\mathbb O} &
\begin{array}{cc}
-1&1\\
0&-1\\ \end{array}&
{\mathbb O}\\
{\mathbb O}&{\mathbb O}&-{\mathbb I_{n-4}}\\
\end{array}
\right),
\end{equation}
while for $n$ odd,
\begin{equation}\label{eq:J-Vodd}
J_0=\left(\begin{array}{cc}
\begin{array}{ccc}
e^{P}&0&0\\
0&e^{-P}&0\\
0&0&1\\ \end{array}&
{\mathbb O}\\
{\mathbb O}&-{\mathbb I_{n-3}}\\
\end{array}
\right),
\end{equation}
where $P=\sum_{i=1}^n Z_i+\sum_{j=1}^{n-3} Y_j$ is the central element corresponding to the face of the fat-graph.
\end{theorem}

\begin{theorem}\label{thm:JordanCFP}
Let $S$ be an upper triangular matrix with $1$ on the diagonal and off diagonal entries
$$
S_{ij}= -\Tr(\gamma_i \gamma_j),\qquad i<j,
$$
with $\gamma_1,\dots,\gamma_n$ given in terms of shear coordinates
by formula (\ref{eq:basisnn}). Then  the matrix $S^{-T}S$, has the
following Jordan form  for $n$ even:
\begin{equation}\label{eq:J-CFP}
J_0=\left(\begin{array}{cc}
\begin{array}{cccc}
-e^{P_1}&0&0&0\\
0&-e^{-P_1}&0&0\\
0&0&-e^{P_2}&0\\
0&0&0&-e^{-P_2}\\
\end{array}& {\mathbb O}\\
{\mathbb O}&-{\mathbb I_{n-4}}\\
\end{array}
\right),
\end{equation}
and for $n$ odd:
\begin{equation}\label{eq:CFP1}
J_0=\left(\begin{array}{cc}
\begin{array}{ccccc}
e^{P}&0&0&0&0\\
0&e^{-P}&0&0&0\\
0&0&-1&1&0\\
0&0&0&-1&0\\
0&0&0&0&1\\
\end{array}& {\mathbb O}\\
{\mathbb O}&-{\mathbb I_{n-5}}\\
\end{array}
\right),
\end{equation}
where ${\mathbb I_{n-4}}$ and ${\mathbb I_{n-5}}$ are respectively
the $(n-4)\times(n-4)$ and   $(n-5)\times(n-5)$ identity matrices
and  $P_1=\sum_{i=1}^n Z_i$ ,$P_2=\sum_{j=1}^{2n-6}Y_j$
are the perimeters of the $2$ holes in the case of
$n$ even, and $P=\sum_{i=1}^n Z_i+\sum_{j=1}^{2n-6}Y_j$
is the perimeter of the one hole for $n$ odd.
\end{theorem}

\begin{remark}
Very similar Jordan normal forms appear for the matrix $S^{-T}S$
where $S$ is the Stokes matrix associated to the Frobenius manifold
structure on Hurwitz space (see Theorem 4 in \cite{SV}). However,
in that case the central elements $P$ or $P_1,P_2$ are rational
multiples of $2\pi i$ rather than real numbers.
\end{remark}

A first step in the direction of proving Theorems \ref{thm:JordanAn} and \ref{thm:JordanCFP}
is carried out in the next Lemma:

\begin{lm}\label{lem-rank}
The matrix of the symmetric form $G_{ij}=(S^{T}+S)_{ij}$ has at most rank four
in the case of a Riemann surface $\Sigma_{g,s,0}$ of genus $g$ with $s=1,2$ and at most rank three in the $\mathcal A_n$ case.
\end{lm}

\proof We prove this lemma in the next subsection where an
interesting interpretation in terms of $n$ particle model in
Minkowski space is studied. \endproof

\subsection{Minkowski space model}\label{ss:Min}

In both $CFP$ and $\mathcal A_n$ cases, each element $G_{ij}$ can be
presented as $G_{ij}=-\tr \gamma_i \gamma_j$, where $\gamma_k$,
$k=1,\dots,n$, are given by (\ref{eq:basisn}) for the $\mathcal A_n$ case
and, thanks to Remark \ref{rk:g},  by the following matrices in $CFP$ case:
\begin{eqnarray}
&{}&
\gamma_1=F,\nn\\
&{}&
\gamma_2 =- X_{\frac{Z_1}{2}} L X_{Z_2} R X_{\frac{Z_1}{2}}\nn\\
&&
\gamma_3= -X_{\frac{Z_1}{2}} R X_{Y_{n-2}} L X_{Z_3} R X_{Y_1} L X_{\frac{Z_1}{2}}   \nn\\
&{}&
\dots\nn \\
&{}&
\gamma_i =- X_{\frac{Z_1}{2}} R X_{Y_{n-2}} R X_{Y_{n-1}}  \dots R X_{Y_{n+i-5}} L X_{Z_i} R X_{Y_{i-2}} L  \dots
X_{Y_1} L X_{\frac{Z_1}{2}},\label{eq:basisnn}\\
&{}&
\dots \nn \\
&{}&
\gamma_{n-1} =- X_{\frac{Z_1}{2}} R X_{Y_{n-2}} R X_{Y_2}  \dots R X_{Y_{2n-6}} L X_{Z_{n-1}}  R X_{Y_{n-3}} L  \dots
X_{Y_1} L X_{\frac{Z_1}{2}},\nn \\
&{}&
\gamma_n =-  X_{\frac{Z_1}{2}} R X_{Y_{n-2}} R X_{Y_{n-1}}  \dots R X_{Y_{2n-6}} R X_{Z_n} R X_{Y_{n-3}} L  \dots X_{Y_1} L X_{\frac{Z_1}{2}}.\nn
\end{eqnarray}
Expand  $\gamma_1,\dots,\gamma_n$ as
$$\gamma_i=\sum_{\alpha=1}^4 v_\alpha^{(i)}\sigma_\alpha,\qquad i=1,\dots,n,
$$
where $\sigma_1,\dots,\sigma_4$ are the real Pauli matrices
$$
\sigma_4=\frac{1}{\sqrt{2}}\left(
\begin{array}{cc}
1 & 0 \\
0 & 1 \\
\end{array}
\right),\quad
\sigma_3=\frac{1}{\sqrt{2}}\left(
\begin{array}{cc}
1 & 0 \\
0 & -1 \\
\end{array}
\right),
$$
$$
\sigma_1=\frac{1}{\sqrt{2}}\left(
\begin{array}{cc}
0 & 1 \\
-1 & 0 \\
\end{array}
\right), \quad
\sigma_2=\frac{1}{\sqrt{2}}\left(
\begin{array}{cc}
0 & 1 \\
1 & 0 \\
\end{array}
\right),
$$
in the $CFP$ case. In the latter case we have:
\be
G_{ij}=v_1^{(i)}v_1^{(j)}-v_2^{(i)}v_2^{(j)}-v_3^{(i)}v_3^{(j)}-v_4^{(i)}v_4^{(j)}
=\sum_{\alpha,\beta=1}^4 v_{\alpha}^{(i)}v_{\beta}^{(j)}\eta^{\alpha\beta},
\label{Min-NR}
\ee
where
\be
\eta^{\alpha\beta}=\hbox{diag\,}(+,-,-,-)
\label{Min-ten}
\ee
is the metric tensor of the Minkowski $3+1$-dimensional space--time.

In the $\mathcal A_n$ case, because each $\gamma_i$ is a conjugate of $F$, $\tr \gamma_i=0$, no fourth component occurs. We then have
\be
G_{ij}=v_1^{(i)}v_1^{(j)}-v_2^{(i)}v_2^{(j)}-v_3^{(i)}v_3^{(j)}
=\sum_{\alpha,\beta=1}^3 v_{\alpha}^{(i)}v_{\beta}^{(j)}\eta^{\alpha\beta},
\label{Min-orbi}
\ee
and $\eta^{\alpha\beta}=\hbox{diag\,}(+,-,-)$ is here the metric tensor of the Minkowski
$2+1$-dimensional space--time.

This concludes the proof of Lemma \ref{lem-rank}. \hfill{$\square$}

\begin{remark}
It is interesting to notice that in the both cases, we can therefore
associate $v^{(i)}_\alpha$ with the components of $4$- or
$3$-dimensional vector ${\mathbf v}^{(i)}$ in the corresponding
Minkowski space. Due to the fact that $\Tr\gamma_i^2=2$ we obtain
the restriction
\be
\|{\mathbf v}^{(i)}\|^2\equiv v^{\alpha\,(i)}v^{(i)}_\alpha=({\mathbf v}^{(i)},{\mathbf v}^{(i)})=2 \ \forall i,
\ee
where we  use the standard  repeated indices summation. This implies
that  all the vectors ${\mathbf v}^{(i)}$, $i=1,\dots,n$  lie in the
upper sheet of the hyperboloid of two sheets (they are time-like
vectors in the physical terminology). In this case  $G_{ij}$ is the
scalar product of the corresponding vectors,
\be
G_{ij}=({\mathbf v}^{(i)},{\mathbf v}^{(j)})
\label{Min-G}
\ee
and since the difference of two different time-like vectors lying on the same sheet is a space-like vector
with negative norm, $\|{\mathbf v}^{(i)}-{\mathbf v}^{(j)}\|^2=4-2G_{ij}<0$, and all $G_{ij}$ are greater than
two, as expected.
\end{remark}

\subsection{Proof of Theorems \ref{thm:JordanAn} and \ref{th:main0} in the $\mathcal A_n$ case.}\label{suse:proof1}

We have proved in Lemma \ref{lem-rank} that
$$
\rm{rk}({\mathcal S}^{-T}{\mathcal S}+\mathbb I)=3,
$$
so we only need to compute the remaining $3$ eigenvalues in order to prove Theorem
\ref{thm:JordanAn}. The proof is based on the following lemma:

\begin{lm}\label{lm:mod}
 All the eigenvalues of the matrix ${\mathcal S}^{-T}{\mathcal S}$ are functions of the only modular invariant
parameter $P=\sum_{\alpha=1}^n Z_\alpha+\sum_{\beta=1}^{n-3} Y_\beta$, which is the sum of all the
Teichm\"uller space variables.
\end{lm}

\proof This is a simple consequence of the fact that  the
eigenvalues of the matrix ${\mathcal S}^{-T}{\mathcal S}$  are part
of the monodromy data of the system (\ref{irreg}) and therefore they
must be central elements in the Dubrovin--Ugaglia bracket and
therefore of the Goldman bracket (\ref{eq:Poisson}). As a
consequence the determinant of any linear combination
$\lambda^{-1}S^{T}+\lambda S$ is a modular-invariant function. On
the other hand, this determinant is a Laurent polynomial of order
not higher than $n$ in each of $e^{Z_\alpha/2}$, which inevitably
means that this determinant is a Laurent polynomial of order not
higher than $n$ of $e^{P/2}$ alone. \endproof

The idea of the proof of Theorem \ref{thm:JordanAn} is to use the
modular invariance to choose in a special way the parameters $Z_i$,
$i=2,\dots,n-1$ and $Y_j$, $j=2,\dots,n-1$, leaving $Z_1$ and $Z_n$
arbitrary. In fact, since the eigenvalues  are modular invariants,
if we change some of $Z_i$ and $Y_j$ by preserving their total sum,
the  eigenvalues must remain the same.

Because the determinant of $\lambda^{-1}{\mathcal S}^{T}+\lambda{\mathcal S}$
is a rational function in $e^{Z_\alpha/2}$ and $e^{Y_\beta/2}$, it has a unique analytic continuation in the
domain of {\em complex values of} $Z_1,\dots,Z_n,Y_1,\dots,Y_{n-3}$. The value of the
determinant must then be conserved provided the exponential
$e^P$, $P=\sum_{\alpha=1}^n Z_\alpha+\sum_{\beta=1}^{n-3} Y_\beta$, remains invariant. We now present a
convenient choice of these, complex, parameters. We take the
representation graph (the spine) of the form depicted in Fig.~\ref{comb}, in which we specially indicated geodesic functions that will
play an important role in the proof.

\begin{figure}[h]
{\psset{unit=0.7}
\begin{pspicture}(-7,-2.5)(7,2.5)
\newcommand{\PAT}{%
\pcline(-0.87,0.5)(-0.87,2)
\pcline(0.87,0.5)(0.87,2)
\pscircle*(0,2){0.1}
}
\newcommand{\PATC}{%
\pcline(0,0.87)(4.5,0.87)
\pcline(0,-0.87)(1.4,-0.87)
\pcline(3.1,-0.87)(4.5,-0.87)
\pcline(1.4,-0.87)(1.4,-2)
\pcline(3.1,-0.87)(3.1,-2)
\pscircle*(2.25,-2){0.1}
}
\newcommand{\GEOD}[1]{% color
\psarc[linewidth=1.5pt,linestyle=dashed,linecolor=#1](0,2){0.6}{0}{180}
\psarc[linewidth=1.5pt,linestyle=dashed,linecolor=#1](1.73,-1){0.4}{-120}{60}
\psbezier[linewidth=1.5pt,linestyle=dashed,linecolor=#1](0.6,2)(0.4,.2)(.4,.2)(1.93,-0.65)
\psbezier[linewidth=1.5pt,linestyle=dashed,linecolor=#1](-0.6,2)(-0.6,0.5)(.53,-.8)(1.53,-1.35)
}
\newcommand{\GEODINVERSE}[1]{% color
\psarc[linewidth=1.5pt,linestyle=dashed,linecolor=#1](0,2){0.6}{0}{180}
\psarc[linewidth=1.5pt,linestyle=dashed,linecolor=#1](-1.73,-1){0.4}{120}{300}
\psbezier[linewidth=1.5pt,linestyle=dashed,linecolor=#1](-0.6,2)(-0.4,.2)(-.4,.2)(-1.93,-0.65)
\psbezier[linewidth=1.5pt,linestyle=dashed,linecolor=#1](0.6,2)(0.6,0.5)(-.53,-.8)(-1.53,-1.35)
}
\rput{30}(-5.5,0){\PAT}
\rput{150}(-5.5,0){\PAT}
\rput{150}(-5.5,0){\GEOD{red}}
\rput{-30}(5.5,0){\PAT}
\rput{-150}(5.5,0){\PAT}
\rput{-150}(5.5,0){\GEODINVERSE{green}}
\rput{30}(-5.5,0){
\psarc[linewidth=1.5pt,linestyle=dashed,linecolor=blue](0,2){0.6}{0}{180}
%\psarc[linewidth=1.5pt,linestyle=dashed,linecolor=blue](1.73,-1){0.4}{-120}{60}
\psbezier[linewidth=1.5pt,linestyle=dashed,linecolor=blue](0.6,2)(0.4,.2)(.4,.2)(1.93,-0.65)
\psbezier[linewidth=1.5pt,linestyle=dashed,linecolor=blue](-0.6,2)(-0.6,0.5)(.53,-.8)(1.53,-1.35)
}
\rput{-30}(5.5,0){
\psarc[linewidth=1.5pt,linestyle=dashed,linecolor=blue](0,2){0.6}{0}{180}
%\psarc[linewidth=1.5pt,linestyle=dashed,linecolor=blue](-1.73,-1){0.4}{120}{300}
\psbezier[linewidth=1.5pt,linestyle=dashed,linecolor=blue](-0.6,2)(-0.4,.2)(-.4,.2)(-1.93,-0.65)
\psbezier[linewidth=1.5pt,linestyle=dashed,linecolor=blue](0.6,2)(0.6,0.5)(-.53,-.8)(-1.53,-1.35)
}
\rput(-5,0){\PATC}
\rput(0.5,0){\PATC}
\pcline[linewidth=2pt,linestyle=dashed, linecolor=blue](-3.5,.4)(-.5,.4)
\pcline[linewidth=2pt,linestyle=dashed, linecolor=blue](-3.5,-.4)(-.5,-.4)
\pcline[linewidth=2pt,linestyle=dashed, linecolor=blue](3.5,.4)(.5,.4)
\pcline[linewidth=2pt,linestyle=dashed, linecolor=blue](3.5,-.4)(.5,-.4)
\rput(0,0.87){\makebox(0,0)[cc]{\Large$\cdots$}}
\rput(0,-0.87){\makebox(0,0)[cc]{\Large$\cdots$}}
\rput(0,0.4){\makebox(0,0)[cc]{\Large\color{blue}$\cdots$}}
\rput(0,-0.4){\makebox(0,0)[cc]{\Large\color{blue}$\cdots$}}
\rput(-7.5,2){\makebox(0,0){$Z_{1}$}}
\rput(-7.5,-2){\makebox(0,0){$Z_{2}$}}
\rput(-4.6,1.2){\makebox(0,0)[cc]{$Y_1$}}
\rput(-4.1,-2){\makebox(0,0){$Z_{3}$}}
\rput(-1,1.2){\makebox(0,0)[cc]{$Y_2$}}
%A4
\rput(7.6,2){\makebox(0,0){$Z_{n}$}}
\rput(7.7,-2){\makebox(0,0){$Z_{n-1}$}}
\rput(4.5,1.2){\makebox(0,0)[cc]{$Y_{n-3}$}}
\rput(4.3,-2){\makebox(0,0){$Z_{n-2}$}}
\end{pspicture}
}
\caption{\small{The fat graph for $\mathcal A_n$. The blue geodesic is $G_{1,n}$, the green one is $G_{n-1,n}$ and the red one is $G_{1,2}$.}}
\label{comb}
\end{figure}
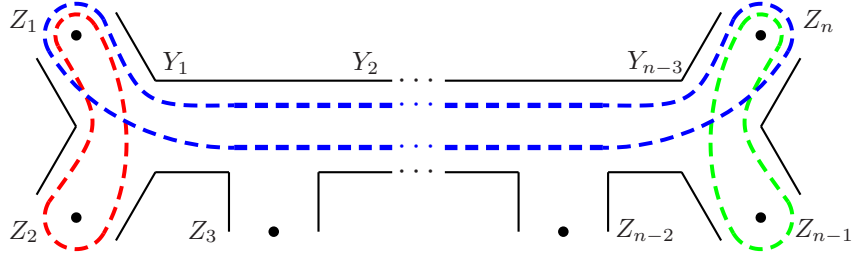

We  choose {\em all} the $Y_j$ to be $-i\pi$, then $X_{Y_j}=\left(
                                                        \begin{array}{cc}
                                                          0 & i \\
                                                          i & 0 \\
                                                        \end{array}
                                                      \right)$
for any $j=1,\dots,n-3$. This special matrix is characterised by that $LX_YL=RX_YR=X_Y$. We also use extensively that
$R=-L^2$ and $L=R^2$. We next take
$Z_2=-Z_3=Z_4=\cdots=(-1)^{n-1}Z_{n-1}$ and leave $Z_1$ and $Z_n$ arbitrary.
Under this choice of the parameters, the entries ${\tilde G}_{ij}:=S_{ij}+S_{ji}$ simplify considerably. Namely, we obtain
\bea
{\tilde G}_{1,2}&=&-{\tilde G}_{13}={\tilde G}_{14}=\cdots=(-1)^{n-1}{\tilde G}_{1,n-1}=\tr LX_{2Z_2}RX_{2Z_1}
\nonumber
\\
{\tilde G}_{i,j}&=&(-1)^{i-j}2,\qquad 1<i\le j<n,
\nonumber
\\
{\tilde G}_{n-1,n}&=&-{\tilde G}_{n-2,n}={\tilde G}_{n-3,n}=\cdots=(-1)^{n-1}{\tilde G}_{2,n}=\tr LX_{2Z_n}RX_{2Z_{n-1}}
\nonumber
\\
{\tilde G}_{1,n}&=&\left\{\begin{array}{ll}
          e^{Z_n+Z_1}+e^{-Z_n-Z_1}, & \hbox{\ even $n$} \\
          \tr LX_{2Z_1}RX_{2Z_{n}}  & \hbox{\ odd $n$} \\
          \end{array}\right.
\label{Gspecial}
\eea
All the entries of the matrix ${\mathcal S}$ are either $\pm 2$, or
$\pm {\tilde G}_{1,2}$, or $\pm {\tilde G}_{n-1,n}$ or ${\tilde G}_{1,n}$. We are now going to show that we can re-arrange the rows
and columns of the matrix  $\lambda {\mathcal S}+\lambda^{-1}{\mathcal S}^T$ in order to obtain the form:
\be\det(\lambda {\mathcal S}+\lambda^{-1}{\mathcal S}^T)=
\left[
                            \begin{array}{cccccc}
                       \lambda+\lambda^{-1} & \lambda c & \lambda a & \lambda a & \cdots & \lambda a \\
                       \lambda^{-1} c & \lambda+\lambda^{-1} & \lambda b & \lambda b & \cdots & \lambda b \\
                       \lambda^{-1}a & \lambda^{-1}b & \lambda+\lambda^{-1} & 2\lambda & \cdots & 2\lambda \\
                    \lambda^{-1}a & \lambda^{-1}b & 2\lambda^{-1} & \lambda+\lambda^{-1} & \ddots & \vdots \\
                       \vdots & \vdots & \vdots & \ddots & \ddots & 2\lambda \\
             \lambda^{-1}a & \lambda^{-1}b & 2\lambda^{-1} & \cdots & 2\lambda^{-1} & \lambda+\lambda^{-1} \\
                            \end{array}
                          \right].
\label{abc-matrix}
\ee
For such matrix form (\ref{abc-matrix}) we can easily compute the determinant:
\bea
\det(\lambda {\mathcal S}+\lambda^{-1}{\mathcal S}^T)&=&
[(\lambda+\lambda^{-1})^2-c^2](\lambda-\lambda^{-1})^2I_{n-4}
\nonumber\\
&{}&+(\lambda+\lambda^{-1})[(\lambda+\lambda^{-1})^2+abc-a^2-b^2-c^2]I_{n-3},
\eea
where $I_k$ is $(\lambda-\lambda^{-1})^k$ times the determinant of the skew--symmetric matrix with all the
entries above the diagonal equal to the unity; this determinant is zero for odd $k$ and $1$ for even $k$. So, we obtain
\be\label{eq:det:abc}
\det(\lambda {\mathcal S}+\frac{{\mathcal S}^T}{\lambda})=\left\{
\begin{array}{ll}
[(\lambda+\lambda^{-1})^2-c^2](\lambda-\lambda^{-1})^{n-2}, & \hbox{even $n$},\\
\\
(\lambda+\lambda^{-1})[(\lambda+\lambda^{-1})^2 +abc -a^2-b^2-c^2](\lambda-\lambda^{-1})^{n-3}, &
\hbox{odd $n$}.
\end{array}\right.
\ee
Let us prove formula (\ref{abc-matrix}) and deduce the values of the eigenvalues of $J_0$ in the even and in the odd dimensional cases separately.

\vskip 2mm
\noindent {\em For even $n$}, we have that $\det(\lambda S+\lambda^{-1} S^T)$ is given by
$$
{\small{\left|
                              \begin{array}{ccccccc}
\lambda+\lambda^{-1} & \lambda {\tilde G}_{1,2} & -\lambda {\tilde G}_{1,2}
 & \lambda  {\tilde G}_{1,2} &\dots& -\lambda  {\tilde G}_{1,2} &
\lambda  {\tilde G}_{1,n} \\
\lambda^{-1} {\tilde G}_{1,2} & \lambda+\lambda^{-1} & -2\lambda & 2\lambda &\dots & -2\lambda & -\lambda  {\tilde G}_{n-1,n} \\
-\lambda^{-1} {\tilde G}_{1,2} & -2\lambda^{-1} & \lambda+\lambda^{-1} & -2\lambda&\dots & 2\lambda & \lambda  {\tilde G}_{n-1,n} \\
\dots&\dots&\dots&\dots&\dots&\dots&\dots\\
\lambda^{-1} {\tilde G}_{1,2} & 2\lambda^{-1} &\dots&- 2\lambda^{-1}
& \lambda+\lambda^{-1} & -2\lambda & -\lambda  {\tilde G}_{n-1,n} \\
-\lambda^{-1} {\tilde G}_{1,2} & -2\lambda^{-1} & \dots& 2\lambda^{-1} & -2\lambda^{-1}
& \lambda+\lambda^{-1} & \lambda  {\tilde G}_{n-1,n} \\
\lambda^{-1} {\tilde G}_{1,n} & -\lambda^{-1} {\tilde G}_{n-1,n} &\dots& \lambda^{-1} {\tilde G}_{n-1,n} & -\lambda^{-1}G_{n-1,n} &
\lambda^{-1} {\tilde G}_{n-1,n} & \lambda+\lambda^{-1} \\
                              \end{array}
                            \right|}}
$$
and multiplying the odd columns and rows by $-1$ and cyclically
permuting rows and columns $\{1,2,\dots,n-1,n\}\to\{n,1,2,\dots,
n-1\}$, we obtain the matrix of the form (\ref{abc-matrix}) with
$a=-G_{n-1,n}$, $b=G_{1,2}$, and $c=G_{1,n}$. Neither $a$ nor $b$
however contribute to the determinant (\ref{eq:det:abc}) for even
$n$, whereas, from (\ref{Gspecial}), $G_{1,n}=e^P+e^{-P}$ (because
the contribution from other $Z_i$ vanish for even $n$,
$Z_2+\cdots+Z_{n-1}=0$). For even $n$ we therefore have
\be \det
(\lambda{\mathcal S}+\lambda^{-1}{\mathcal S}^T)=
[(\lambda+\lambda^{-1})^2-(e^P+e^{-P})^2](\lambda-\lambda^{-1})^{n-2},
\ee
and the roots of the characteristic equation $\det({\mathcal
S}^{-T}{\mathcal S}-\eta)=0$ ($\eta=-\lambda^2$) are
$\eta=\{-e^P,-e^{-P},-1,\dots,-1\}$. Since the rank of ${\mathcal
S}^{-T}{\mathcal S}+1$ is less or equal three, the Jordan form (in
the case of nonzero $P$) must have $n-2$  $1\times 1$ blocks
corresponding to the eigenvalues: $-e^P$, $-e^{-P}$, and $n-4$
eigenvalues $-1$,  and one $2\times 2$ block $\left(
                                                                   \begin{array}{cc}
                                                                     -1 & 1 \\
                                                                     0 & -1 \\
                                                                   \end{array}
                                                                 \right)$.
This concludes the proof of (\ref{eq:J-Veven}) for $n$ even.

\vskip 2mm
\noindent{\em For odd $n$}  we have   that $\det(\lambda S+\lambda^{-1} S^T)$ is given by
$$
\small{
\left|
                              \begin{array}{cccccccc}
\lambda+\lambda^{-1} & \lambda {\tilde G}_{1,2} & -\lambda {\tilde G}_{1,2} & \lambda {\tilde G}_{1,2} & -\lambda G_{1,2} &
\lambda {\tilde G}_{1,2} &\lambda {\tilde G}_{1,n} \\
\lambda^{-1}{\tilde G}_{1,2} & \lambda+\lambda^{-1} & -2\lambda & 2\lambda & -2\lambda & 2\lambda & \lambda {\tilde G}_{n-1,n} \\
-\lambda^{-1}{\tilde G}_{1,2} & -2\lambda^{-1} & \lambda+\lambda^{-1} & -2\lambda & 2\lambda & -2\lambda
& -\lambda {\tilde G}_{n-1,n} \\
\dots& \dots&\dots&\dots&\dots&\dots&\dots\\
-\lambda^{-1}{\tilde G}_{1,2} & -2\lambda^{-1} &\dots & 2\lambda^{-1} & -2\lambda^{-1}
& -2\lambda & -\lambda {\tilde G}_{n-1,n} \\
\lambda^{-1}{\tilde G}_{1,2} & 2\lambda^{-1} &\dots & -2\lambda^{-1} & 2\lambda^{-1}
& \lambda+\lambda^{-1} &
\lambda {\tilde G}_{n-1,n} \\
\lambda^{-1}{\tilde G}_{1,n} & \lambda^{-1}{\tilde G}_{n-1,n} &\dots& -\lambda^{-1}{\tilde G}_{n-1,n} & \lambda^{-1}{\tilde G}_{n-1,n}&
\lambda^{-1}{\tilde G}_{n-1,n} & \lambda+\lambda^{-1} \\
                              \end{array}
                            \right|}
$$
and multiplying the odd columns and rows by $-1$ and cyclically permuting rows and
columns $\{1,2,\dots,n-1,n\}\to\{n,1,2,\dots, n-1\}$, we obtain the matrix of the form (\ref{abc-matrix})
with $a={\tilde G}_{n-1,n}$, $b=G_{1,2}$, and $c=G_{1,n}$. Now, the elements $G_{1,2}$, ${\tilde G}_{n-1,n}$, and $G_{1,n}$
(see their explicit expressions in (\ref{Gspecial})) constitute the Markov triple, that is,
$abc-a^2-b^2-c^2=(e^P-e^{-P})^2$, where $P=Z_1+Z_2+Z_n$ is the perimeter of the hole (the
remaining $Z_i$ are mutually canceled).
For odd $n$, we therefore have
\be
\det (\lambda{\mathcal S}+\lambda^{-1}{\mathcal S}^T)=
[(\lambda+\lambda^{-1})^2+(e^P-e^{-P})^2](\lambda+\lambda^{-1})(\lambda-\lambda^{-1})^{n-3},
\ee
and the roots of the characteristic equation
$\det({\mathcal S}^{-T}{\mathcal S}-\eta)=0$ ($\eta=-\lambda^2$)
are now $\eta=\{e^P,e^{-P},1,-1,\dots,-1\}$.
Since the rank of ${\mathcal S}^{-T}{\mathcal S}+1$ is less or
equal three, all these numbers are eigenvalues (for $P\ne0$) and
the Jordan form is diagonal. This concludes the proof of (\ref{eq:J-Vodd}) for $n$ odd.
\endproof

\subsubsection{Symplectic leaves corresponding to $\mathcal A_n$}

We are now ready to prove that the dimension of the symplectic leaves $\mathcal L_{\mathcal A_n}$ corresponding to $\mathcal A_n$ is
$$
\dim_{\mathbb C}\left( L_{\mathcal A_n}\right)=2(n-2)
$$
which is the double the real dimension of the Teichm\"uller space.

\proof Thanks to Theorem \ref{thm:JordanAn}, for $n$ even,
$$
J_0= n_{\lambda,1} \left(J_{\lambda,1} \oplus J_{\frac{1}{\lambda},1} \right)  +
m_{-1,2} J_{-1,2}  +n_{-1,1}\left(J_{-1,1} \oplus J_{{-1},1} \right)
$$
where
$$
n_{\lambda,1} =1,\qquad m_{-1,2}=1,\qquad n_{-1,1}=\frac{n-4}{2}.
$$
while for  $n$  odd,
$$
J_0= n_{\lambda,1} \left(J_{\lambda,1} \oplus J_{\frac{1}{\lambda},1} \right)  +
 m_{1,1} J_{1,1} +n_{-1,1}\left(J_{-1,1} \oplus J_{{-1},1} \right)
$$
where
$$
n_{\lambda,1} =1,\qquad m_{1,1}=1,\qquad n_{-1,1}=\frac{n-3}{2}.
$$
Using (\ref{eq:bondal}) we get precisely
$$
{\rm d}(S)=\frac{8-5n+n^2}{2}=\frac{n(n-1)}{2}-2(n-2).
$$
This concludes the proof of Theorem \ref{th:main0} in the $\mathcal A_n$ case.
\endproof

\subsection{Proof of Theorems \ref{thm:JordanCFP} and \ref{th:main0} in  the $CFP$ case}\label{suse:proof2}

The idea of the proof is the same as for Theorem
\ref{thm:JordanAn}. We already proved  in Lemma \ref{lem-rank} that
$\rm{rk}({\mathcal S}^{-T}{\mathcal S}+\mathbb I)=4$, so we only
need to compute the remaining $4$ eigenvalues. Lemma \ref{lm:mod} is
still valid and we will now show how to pick the parameters
$Z_i$,and $Y_j$, in a way to simplify computations.

\vskip 2mm
\noindent{\em{Odd $n$.}} In this case, we have just one
hole and we can set $e^{Y/2}=i$ and
$Z_1=-Z_2=Z_3=\cdots=-Z_{n-1}=Z_n$ as before. Then,
$G_{ij}=(-1)^{i-j}2$ for $i,j\ne 1,n$, $G_{1,2}=(-1)^jG_{1,j}=\tr
LX_{-Z_1}RX_{Z_1}$ for $2\le j\le n-1$, $G_{n-1,n}=(-1)^jG_{j,n}=\tr
LX_{Z_1}RX_{-Z_1}$ for $2\le j\le n-1$, and by multiplying the odd
columns and rows by $-1$ and cyclically permuting rows and columns
$\{1,2,\dots,n-1,n\}\to\{n,1,2,\dots, n-1\}$, we obtain the matrix
(\ref{abc-matrix}) with $a=G_{n-1,n}$, $b=G_{1,2}$, and $c=G_{1,n}$
having the same determinant (\ref{eq:det:abc}).

The characteristic equation  $\det\left(S^{-T}{S}-\eta \mathbb I\right)=0$ where $\eta=-\lambda^2$,
has $(n-3)$-fold root $\eta=-1$ and three single roots $\varphi=1$, $\varphi=e^{P}$, and $\varphi=e^{-P}$.
Then, since the rank of the matrix ${\mathcal S}^{-T}{\mathcal S}+\mathbb I)$ is four in the $CFP$ case, we obtain formula (\ref{eq:CFP1}).

\vskip 2mm
\noindent{\em{Even $n$.}}
In this  case, we have two modular-invariant parameters $P_1$ and $P_2$ such that
$\sum_{\alpha=1}^n Z_\alpha+\sum Y_{\beta=1}^{2n-6}=P_1+P_2$. We cannot now set all the variables $Y_j$ to be $i\pi$ in the
graph in Fig.~\ref{octopus}  because those are the variables that distinguish between the perimeters
of these two holes. We can set however $Z_1=-Z_2=Z_3=\cdots=-Z_n$, take two of the variables $Y$, say, the
variables $Y_1$ and $Y_{n-2}$ of the two edges (above and below) that separate $Z_2$ and $Z_3$ to be arbitrary and
set all the remaining $Y_j$ to be $i\pi$ (and $X_Y=\left({0\ i \atop i\ 0 }\right)$). We then
have {\em six} basic matrix elements, $a=G_{1,n}$, $b=G_{2,n}$, $c=G_{3,n}$,
$d=G_{1,2}$, $e=G_{1,3}$, and $f=G_{2,3}$, and the matrix $\lambda^{-1}S^{T}
+\lambda S$ reduces by the same operations of row/column multiplication by $-1$ and
cyclic permutations of row/columns to the form
\bea
&&\det\left(
  \begin{array}{ccccccc}
    \lambda+\lambda^{-1} & a\lambda & b\lambda & c\lambda & c\lambda & \cdots & c\lambda \\
    a\lambda^{-1}  & \lambda+\lambda^{-1} & d\lambda  & e\lambda & e\lambda & \cdots & e\lambda \\
    b\lambda^{-1}  &d\lambda^{-1}  & \lambda+\lambda^{-1} & f\lambda  & f\lambda &  \cdots & f\lambda \\
    c\lambda^{-1}  &e\lambda^{-1}  & f\lambda^{-1} & \lambda+\lambda^{-1}  & 2\lambda & \cdots & 2\lambda \\
    c\lambda^{-1}  &e\lambda^{-1}  & f\lambda^{-1} & 2\lambda^{-1}  & \lambda+\lambda^{-1}  & \ddots & \vdots \\
    \vdots & \vdots & \vdots & \vdots & \ddots & \ddots & 2\lambda \\
    c\lambda^{-1} & e\lambda^{-1} & f\lambda^{-1}  & 2\lambda^{-1}  & \cdots & 2\lambda^{-1} & \lambda+\lambda^{-1}  \\
  \end{array}
\right)
\nonumber
\\
&&\nn\\
&&\qquad=(\lambda-\lambda^{-1})^{n-4}
\det\left(
  \begin{array}{cccc}
    \lambda+\lambda^{-1} & a\lambda & b\lambda & c\lambda\\
    a\lambda^{-1}  & \lambda+\lambda^{-1} & d\lambda  & e\lambda \\
    b\lambda^{-1}  &d\lambda^{-1}  & \lambda+\lambda^{-1} & f\lambda  \\
    c\lambda^{-1}  &e\lambda^{-1}  & f\lambda^{-1} & \lambda+\lambda^{-1}
  \end{array}
\right)
\nonumber
\eea
We express the remaining determinant through
two invariant determinants:
$$
D_1\equiv\det\left(
  \begin{array}{cccc}
    2 & a & b & c\\
    a  & 2 & d  & e \\
    b  &d  & 2 & f  \\
    c  &e  & f & 2
  \end{array}
\right)=\bigl(e^{P_1/2}+e^{-P_1/2}-e^{P_2/2}-e^{-P_2/2}\bigr)^2,\ \hbox{at $\lambda=\pm1$}
$$
and
$$
D_2\equiv\det\left(
  \begin{array}{cccc}
    0 & a & b & c\\
    -a  & 0 & d  & e \\
    -b  &-d  & 0 & f  \\
    -c  &-e  & -f & 0
  \end{array}
\right)=\bigl(e^{P_1/2}+e^{-P_1/2}+e^{P_2/2}+e^{-P_2/2}\bigr)^2,\ \hbox{at $\lambda=\pm i$}.
$$
For the determinant in question, we have
\bea
&&\det (\lambda^{-1}S^{T}+\lambda S)
\nonumber
\\
&&\qquad=(\lambda-\lambda^{-1})^{n-4}\bigl((\lambda+\lambda^{-1})^2(\lambda-\lambda^{-1})^2
+\frac{D_1}{4}(\lambda+\lambda^{-1})^2
-\frac{D_2}{4}(\lambda-\lambda^{-1})^2\bigr)
\nonumber
\\
&&\qquad=(\lambda-\lambda^{-1})^{n-4}\biggl((\lambda^2-\lambda^{-2})^2
-4\cosh(P_1/2)\cosh(P_2/2)(\lambda^2+\lambda^{-2})
\biggr.
\nonumber
\\
&&\qquad\qquad\qquad \biggl.+4\cosh^2(P_1/2)+4\cosh^2(P_2/2)\biggr)
\label{even-n-NR}
\eea

The roots of (\ref{even-n-NR}) for $\varphi=-\lambda^2$ are
$n-4$-fold root $\varphi=-1$ and four simple roots
$\varphi=-e^{(P_1+P_2)/2}$, $\varphi=-e^{-(P_1+P_2)/2}$,
$\varphi=-e^{(P_1-P_2)/2}$, and $\varphi=-e^{-(P_1-P_2)/2}$. When
all these roots are distinct, the Jordan form is diagonal and all
the roots correspond to eigenvectors. These completes the analysis of the Jordan forms for the $CFP$ case.\endproof

\subsubsection{Symplectic leaves corresponding to the $CFP$ case}

We can now prove that the dimension of the symplectic leaves $\mathcal L_{CFP}$ corresponding to the $CFP$ case is
$$
\dim_{\mathbb C}\left( L_{\mathcal A_n}\right)= \left\{\begin{array}{lc}
3n-7 &\hbox{ for $n$ odd}\\
3n-8 &\hbox{ for $n$ even}\\
\end{array}\right.
$$
which is the double the real dimension of the Teichm\"uller space.

\proof Thanks to Theorem \ref{thm:JordanCFP}, we have that in the case of $n$ even the Jordan normal form decomposes as:
$$
J_0= n_{\lambda_1,1} \left(J_{\lambda_1,1} \oplus J_{\frac{1}{\lambda_1},1} \right) +
n_{\lambda_2,1} \left(J_{\lambda_2,1} \oplus J_{\frac{1}{\lambda_2},1} \right) +
n_{-1,1}\left(J_{-1,1} \oplus J_{{-1},1} \right)
$$
where
$$
n_{\lambda_1,1} =1,\qquad n_{\lambda_2,1}=1,\qquad n_{-1,1}=\frac{m-4}{2},
$$
and for $n$ odd
 $$
J_0= n_{\lambda,1} \left(J_{\lambda,1} \oplus J_{\frac{1}{\lambda},1} \right)  +
m_{-1,2} J_{-1,2} + m_{1,1} J_{1,1} +n_{-1,1}\left(J_{-1,1} \oplus J_{{-1},1} \right)
$$
where
$$
n_{\lambda,1} =1,\qquad m_{-1,2}=1,\qquad m_{1,1}=1,\qquad n_{-1,1}=\frac{m-5}{2}.
$$
By using (\ref{eq:bondal}) we conclude the proof of Theorem \ref{th:main0} in the $CFP$ case.
\endproof

\subsection{Complexification}\label{suse:c}

In this section we observe that the Stokes matrices belonging to the
Teichm\"uller symplectic leaves $\mathcal L_{\mathcal A_n}$ and $\mathcal
L_{CFP}$ can be parameterized in terms of {\it complex
coordinates}\/ $Z_1,\dots,Z_n$, $Y_1,\dots,Y_k$ where $k=n-3$ in the
$\mathcal A_n$ case and $k=2n-6$ in the $CFP$ case by the same formulae
$$
S_{ij} = - \Tr\gamma_i \gamma_j, \qquad i<j,
$$
where $\gamma_i,\gamma_j$ are now matrices in $SL_2(\mathbb C)$
still given by formulae (\ref{eq:basisn}) and (\ref{eq:basisnn}) with complex $Z_1,\dots,Z_n$, $Y_1,\dots,Y_k$ .

When the coordinates $Z_i$ become complex, we can still use the same parameterization
of elements of the discretely acting group, which becomes now a finitely generated subgroup of
$PSL(2,{\mathbb C})$, not $PSL(2,{\mathbb R})$, i.e., a {\em Kleinian} group.
This Kleinian
group $\Delta_{g'}\subset PSL(2,{\mathbb C})$ describes now a
handlebody, that is, the quotient of the upper half-space ${\mathbb H}^+_3:={\mathbb C}\times {\mathbb R}^+$ by the action
of  $\Delta_{g'}$. The handlebody is geometrically  a filled Riemann surface whose boundary is a
closed Riemann surface of genus $g'=2g+s-1$  obtained from
the action of this group on the boundary of ${\mathbb H}^+_3$, i.e.,
on the complex plane ${\mathbb C}$, and admits a Schottky uniformisation.

Note that in this approach we {\em do not} present the three-dimensional manifold as a direct product of
a Riemann surface (with holes) and a time interval; instead we have an actual handlebody endowed with the set of
closed geodesics inside it; each closed geodesic corresponds, as before, to a conjugacy class of the
Kleinian group.

Note that, in this case, we loose the
distinction between holes and handles of the original Riemann surface $\Sigma_{g,s}$: if we consider two Riemann surfaces $\Sigma_{g_1,s_1}$ and  $\Sigma_{g_2,s_2}$ such that they are described by the same number of shear coordinates, or in other words such that
$\dim\left({\mathcal T}_{g_1,s_1}\times{\mathbb R}^{s_1}\right)=
\dim\left({\mathcal T}_{g_2,s_2}\times{\mathbb R}^{s_2}\right)$, they can be
considered as different parameterisations of
the {\em same} handle--body, as we demonstrate on the example below.

\begin{example}
Complexification of the Teichm\"uller space ${\mathcal T}_{1,1}$ of
a torus with one hole and of the Teichm\"uller space ${\mathcal
T}_{0,3}$ of a sphere with three holes.

In Fig.~\ref{fi:schottky} the original (two-dimensional) Riemann
surface is obtained under the action of a Kleinian group in
${\mathbb H}^+_3$ restricted to the real vertical slice ${\mathbb
H}^+_2$. Of course, this is possible only when the real slice of the
Kleinian group is simultaneously a Fuchsian group itself,
i.e., a discrete subgroup of $PSL(2,{\mathbb R})$. However, we can
continuously vary the parameters $X_i$ in the complex domain to
ensure a smooth transition between two patterns, as shown in
Fig.~\ref{fi:schottky}.
\begin{figure}[h]
\begin{center}
{\psset{unit=0.6}
\begin{pspicture}(-6,-6)(6,6)
%A3
\newcommand{\PAT}{%
\pcline(-6,0)(6,0)
\pcline(-6,0)(-6,4)
\pcline(6,0)(6,4)
\psbezier(-6,4)(-3,3.6)(3,4.4)(6,4)
\pcline(7.5,1.5)(4.5,-1.5)
\pcline(-6,0)(-7.5,-1.5)
\pcline(4.5,-1.5)(-7.5,-1.5)
\pcline(7.5,1.5)(6,1.5)
\pcline[linestyle=dashed](4.5,1.5)(6,1.5)
\pcline[linestyle=dashed](-6,0)(-4.5,1.5)
\rput(-6.3,-1.1){\makebox(0,0){${\mathbb C}$}}
}
\rput(0,2){
\psellipse[linewidth=1.5pt,linecolor=red](-4.3,0)(1,.4)
\psframe[linecolor=white, fillstyle=solid, fillcolor=white](-5.4,0.5)(-3.2,0)
\psellipse[linewidth=1.5pt,linestyle=dashed,linecolor=red](-4.3,0)(1,.4)
\psarc[linecolor=red](-4.3,0){1}{0}{180}
\psellipse[linewidth=1.5pt,linecolor=blue](-1.8,0)(.7,.28)
\psframe[linecolor=white, fillstyle=solid, fillcolor=white](-2.8,0.4)(-1,0)
\psellipse[linewidth=1.5pt,linestyle=dashed,linecolor=blue](-1.8,0)(.7,.28)
\psarc[linecolor=blue](-1.8,0){.7}{0}{180}
\psellipse[linewidth=1.5pt,linecolor=red](1.3,0)(.7,.28)
\psframe[linecolor=white, fillstyle=solid, fillcolor=white](.5,0.4)(2.1,0)
\psellipse[linewidth=1.5pt,linestyle=dashed,linecolor=red](1.3,0)(.7,.28)
\psarc[linecolor=red](1.3,0){.7}{0}{180}
\psellipse[linewidth=1.5pt,linecolor=blue](4,0)(1,.4)
\psframe[linecolor=white, fillstyle=solid, fillcolor=white](2.9,0.5)(5.1,0)
\psellipse[linewidth=1.5pt,linestyle=dashed,linecolor=blue](4,0)(1,.4)
\psarc[linecolor=blue](4,0){1}{0}{180}
\psarc[linewidth=1.5pt,linestyle=dashed,linecolor=blue](1.1,0){2.7}{0}{180}
\psarc[linewidth=1.5pt,linestyle=dashed,linecolor=red](-1.5,0){2.7}{0}{180}
\psarc[linewidth=2pt,linecolor=blue]{<-}(-.25,0){1.1}{50}{130}
\psarc[linewidth=2pt,linecolor=red]{<-}(-.25,0){1.1}{230}{310}
\rput(0,0){\PAT}
\rput(-4.1,1.6){\makebox(0,0){$A$}}
\rput(2.7,1){\makebox(0,0){$B$}}
\rput(1.5,3.5){\makebox(0,0){\small$\Sigma_{1,1}{=}{\mathbb H}^+_2/\Delta_{1,1}$}}
\rput(-5,3){\makebox(0,0){$Z_i$}}
\pscircle[linewidth=0.5pt](-5,3){.6}
}
\pcline[linewidth=2pt]{->}(0,.2)(0,-.6)
\rput(0,-5){
\psellipse[linewidth=1.5pt,linecolor=red](-4.3,0)(1,.4)
\psframe[linecolor=white, fillstyle=solid, fillcolor=white](-5.4,0.5)(-3.2,0)
\psellipse[linewidth=1.5pt,linestyle=dashed,linecolor=red](-4.3,0)(1,.4)
\psarc[linecolor=red](-4.3,0){1}{0}{180}
\psellipse[linewidth=1.5pt,linecolor=red](-1.8,0)(.7,.28)
\psframe[linecolor=white, fillstyle=solid, fillcolor=white](-2.8,0.4)(-1,0)
\psellipse[linewidth=1.5pt,linestyle=dashed,linecolor=red](-1.8,0)(.7,.28)
\psarc[linecolor=red](-1.8,0){.7}{0}{180}
\psellipse[linewidth=1.5pt,linecolor=blue](1.3,0)(.7,.28)
\psframe[linecolor=white, fillstyle=solid, fillcolor=white](.5,0.4)(2.1,0)
\psellipse[linewidth=1.5pt,linestyle=dashed,linecolor=blue](1.3,0)(.7,.28)
\psarc[linecolor=blue](1.3,0){.7}{0}{180}
\psellipse[linewidth=1.5pt,linecolor=blue](4,0)(1,.4)
\psframe[linecolor=white, fillstyle=solid, fillcolor=white](2.9,0.5)(5.1,0)
\psellipse[linewidth=1.5pt,linestyle=dashed,linecolor=blue](4,0)(1,.4)
\psarc[linecolor=blue](4,0){1}{0}{180}
\psarc[linewidth=1.5pt,linestyle=dashed,linecolor=blue](2.7,0){1.2}{0}{180}
\psarc[linewidth=1.5pt,linestyle=dashed,linecolor=red](-3.1,0){1.2}{0}{180}
%\psarc[linewidth=2pt,linecolor=blue]{<-}(-.25,0){1.1}{50}{130}
%\psarc[linewidth=2pt,linecolor=red]{<-}(-.25,0){1.1}{230}{310}
\rput(0,0){\PAT}
\rput(-3.7,1.6){\makebox(0,0){$A'$}}
\rput(2.7,1.6){\makebox(0,0){$B'$}}
\rput(1.5,3.5){\makebox(0,0){\small$\Sigma_{0,3}{=}{\mathbb H}^+_2/\Delta_{0,3}$}}
\rput(-5,3){\makebox(0,0){$X_i$}}
\pscircle[linewidth=0.5pt](-5,3){.6}
}
\end{pspicture}
}
\caption{\small Deformation of the Fuchsian group transforming the torus with one hole into
the sphere with three holes. We let $A$ and $B$ denote the geodesics corresponding to the respective
$A$- and $B$-cycles on $\Sigma_{1,1}$; their images $A'$ and $B'$ are the geodesics corresponding to the
perimeters of two of the holes in $\Sigma_{0,3}$; the perimeter of the third hole is $B'A'$.}
\label{fi:schottky}
\end{center}
\end{figure}
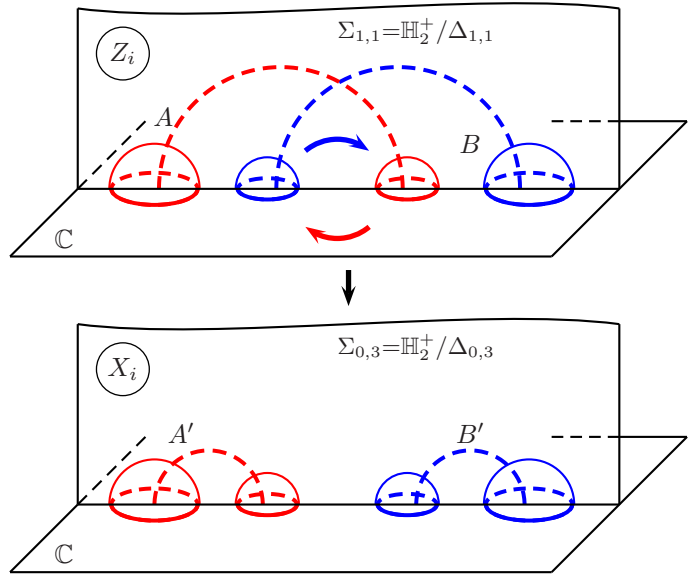

Note that on the intermediate stages of the transition process in
Fig.~\ref{fi:schottky} we have no embedded two-dimensional
(geodesically closed) Riemann surface inside the handlebody; it is
reconstructed only when the group again becomes Fuchsian.

Although the two Riemann surfaces in Fig.~\ref{fi:schottky}, $\Sigma_{1,1}$ and $\Sigma_{0,3}$, have
different topologies, their sets of geodesic lengths are the same, so we say they are {\em isospectral.}

We introduce the set of (decorated) Teichm\"uller space coordinates $Z_i$, $i=1,2,3$,  for ${\mathcal T}_{1,1}$
and $X_i$, $i=1,2,3$, for ${\mathcal T}_{0,3}$; then,
in order for the spectra of geodesic functions to coincide, it suffices to make the identification
(up to the action of the mapping class group in each of the surfaces)
\be
\label{identification}
e^{P_i/2}+e^{-P_i/2}=e^{Z_i/2+Z_{i+1}/2}+e^{-Z_i/2-Z_{i+1}/2}+e^{-Z_i/2+Z_{i+1}/2},\quad i=1,2,3,
\ee
where $P_i=X_i+X_{i+1}$ are the perimeters of three holes of $\Sigma_{0,3}$ and the
(standard) geodesic functions $G_{i,i+1}$ for $\Sigma_{1,1}$ stand in the right-hand sides.

Note that equations (\ref{identification}) not always admit real solutions in terms of $Z_i$ for a given real $X_i$: the
obstruction is provided by the Markov element,
\be
\label{Markov}
{\mathcal M}=G_{1,2}G_{1,3}G_{2,3}-G_{1,2}^2-G_{1,3}^2-G_{2,3}^2.
\ee
In the case of the torus ${\Sigma}_{1,1}$ with real $Z_i$, we have the inequality ${\mathcal M}\ge 0$,
whereas in the case of the sphere ${\Sigma}_{0,3}$ with real $X_i$,
we have the inequality ${\mathcal M}\ge -4$, so Eqs.~\ref{identification} admit real solutions both in
$Z_i$ and in $X_i$ iff ${\mathcal M}\ge 0$.

In Fig.~\ref{fi:correspondence} we depict the explicit relation between the geodesic functions and indicate the
image of the boundary curve. In Fig.~\ref{fi:fatgraph} the same correspondence is presented for the spines
$\Gamma_{1,1}$ and $\Gamma_{0,3}$. Note that neither the intersection indices between the curves nor the Poisson brackets are
preserved under this identification.
\begin{figure}[h]
\begin{center}
{\psset{unit=0.6}
\begin{pspicture}(-8,-3.5)(8,3.5)
%A3
% The next command allows drawing arcs of ellipses with the basepoint (0,0)
\newcommand{\ELLARC}[7]{% color, width, style, angle1, angle2, horizontal half-axis, vertical half-axis
\parametricplot[linecolor=#1, linewidth=#2 pt, linestyle=#3]{#4}{#5}{#6 t cos mul #7 t sin mul}
}
\definecolor{lightblue}{rgb}{.85, .5, 1}
\newcommand{\PATS}[1]{% color
\pscircle[linewidth=1pt,fillstyle=solid,fillcolor=gray](1.8,0){.3}
\pscircle[linewidth=1pt,linecolor=#1](1.8,0){.6}
\psarc[linewidth=1.5pt,linestyle=dashed,linecolor=lightblue](1.8,0){1}{-120}{120}
\pcline[linewidth=1.5pt,linestyle=dashed,linecolor=lightblue](-0.2,0)(1.3,0.86)
\pcline[linewidth=1.5pt,linestyle=dashed,linecolor=lightblue](0.1,-0.17)(1.3,-0.86)
}
\newcommand{\PAT}{%
\rput(-1,-0.3){\ELLARC{black}{1}{solid}{-45}{225}{3.2}{2.5}}
\rput(-1,-0.9){\ELLARC{black}{1}{solid}{20}{160}{.9}{.4}}
\rput(-1,-0.3){\ELLARC{black}{1}{solid}{210}{330}{1.5}{.6}}
\rput(-4.8,1){\makebox(0,0){$Z_i$}}
\pscircle[linewidth=0.5pt](-4.8,1){.6}
\rput(-1,-2.75){\ELLARC{lightblue}{1}{dashed}{0}{180}{1.85}{.5}}
\rput(-1,-2.75){\ELLARC{lightblue}{2}{dashed}{180}{360}{1.85}{.5}}
\rput(-4.1,-2.75){\ELLARC{black}{1}{solid}{-45}{45}{1.2}{.94}}
\rput(2.1,-2.75){\ELLARC{black}{1}{solid}{135}{225}{1.2}{.94}}
}
\rput(-3,1){
\rput(-0.1,0){\ELLARC{green}{1}{dashed}{45}{210}{1.9}{1.9}}
\rput(-0.1,0){\ELLARC{green}{1}{solid}{-150}{45}{1.9}{1.9}}
\rput(-1,-0.6){\ELLARC{blue}{1}{solid}{0}{360}{2}{0.8}}
\rput{-80}(-1,.85){\ELLARC{red}{1}{solid}{180}{360}{1.3}{0.5}}
\rput{-80}(-1,.85){\ELLARC{red}{1}{dashed}{0}{180}{1.3}{0.5}}
\rput(0,0){\PAT}
\rput(-2,1.5){\makebox(0,0){$A$}}
\rput(-3.5,-.6){\makebox(0,0){$B$}}
\rput(1,.7){\makebox(0,0){$BA$}}
\rput(-1,-3.8){\makebox(0,0){$ABA^{-1}B^{-1}$}}
\rput(-4.5,-3){\makebox(0,0){$\Sigma_{1,1}$}}
}
\pcline[linewidth=2pt]{->}(0,0)(1,0)
\rput(4,.2){
\rput{120}(0,0){\PATS{red}}
\rput(0,0){\PATS{blue}}
\rput{-120}(0,0){\PATS{green}}
\rput(-2.5,2){\makebox(0,0){$A'$}}
\pcline[linewidth=0.5pt]{->}(-2.2,1.9)(-1.5,1.7)
\rput(3.6,0){\makebox(0,0){$B'$}}
\pcline[linewidth=0.5pt]{->}(3.2,0)(2.4,0)
\rput(-2.7,-2){\makebox(0,0){$B'A'$}}
\pcline[linewidth=0.5pt]{->}(-2,-1.9)(-1.5,-1.7)
\rput(3.5,-2){\makebox(0,0){$A'B'{A'}^{-1}{B'}^{-1}$}}
\pcline[linewidth=0.5pt]{->}(2,-1.7)(1.8,-1.05)
\rput(0.2,-3.4){\makebox(0,0){$\Sigma_{0,3}$}}
\rput(1.5,2){\makebox(0,0){$X_i$}}
\pscircle[linewidth=0.5pt](1.5,2){.6}
}
\end{pspicture}
}
\caption{\small The transformation between geodesics on $\Sigma_{1,1}$ and $\Sigma_{0,3}$.}
\label{fi:correspondence}
\end{center}
\end{figure}
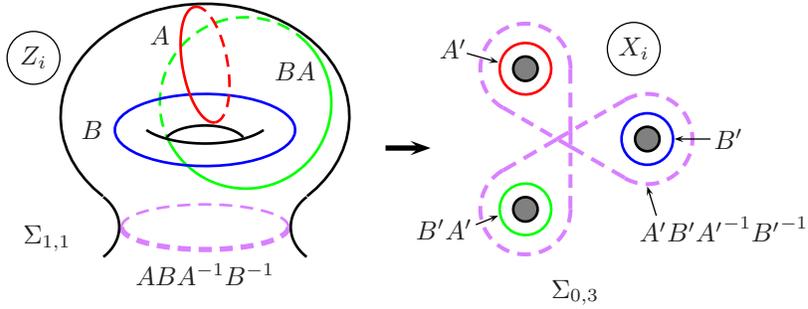

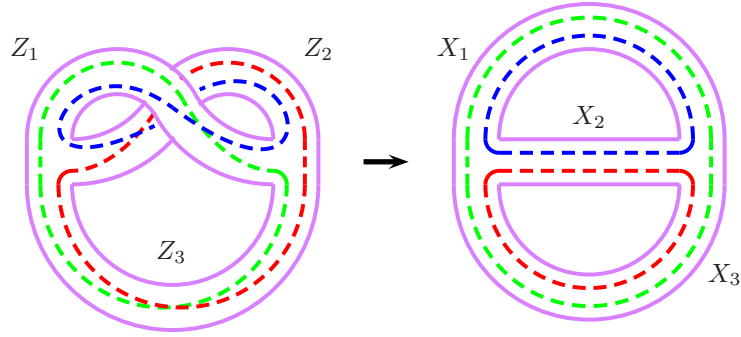
\begin{figure}[h]
\begin{center}
{\psset{unit=0.6}
\begin{pspicture}(-6,-4)(6,3)
%A3
% The next command allows drawing arcs of ellipses with the basepoint (0,0)
\newcommand{\ELLARC}[7]{% color, width, style, angle1, angle2, horizontal half-axis, vertical half-axis
\parametricplot[linecolor=#1, linewidth=#2 pt, linestyle=#3]{#4}{#5}{#6 t cos mul #7 t sin mul}
}
\definecolor{lightblue}{rgb}{.85, .5, 1}
\rput(-3,0.5){
\psarc[linewidth=1.5pt,linecolor=lightblue](0,0){2}{0}{150}
\psarc[linewidth=1.5pt,linecolor=lightblue](0,0){1}{0}{150}
\psarc[linewidth=1.5pt,linestyle=dashed,linecolor=red](0,0){1.7}{0}{150}
\psbezier[linewidth=1.5pt,linestyle=dashed,linecolor=red](-1.47,.85)(-2.17,-.36)(-3,-0.7)(-3.46,-0.7)
\psarc[linewidth=1.5pt,linecolor=lightblue](-3.46,2){2}{270}{330}
\psarc[linewidth=1.5pt,linecolor=lightblue](-3.46,2){3}{270}{330}
\pcline[linewidth=1.5pt,linecolor=lightblue](2,0)(2,-1)
\pcline[linewidth=1.5pt,linestyle=dashed,linecolor=red](1.7,0)(1.7,-1)
}
\rput(-5.46,0.5){
\psarc[linewidth=15pt,linecolor=white](0,0){1.5}{30}{180}
\psarc[linewidth=1.5pt,linecolor=lightblue](0,0){2}{30}{180}
\psarc[linewidth=1.5pt,linecolor=lightblue](0,0){1}{30}{180}
\psarc[linewidth=1.5pt,linestyle=dashed,linecolor=green](0,0){1.7}{30}{180}
\psarc[linewidth=15pt,linecolor=white](3.46,2){2.5}{210}{270}
\psarc[linewidth=1.5pt,linecolor=lightblue](3.46,2){2}{210}{270}
\psarc[linewidth=1.5pt,linecolor=lightblue](3.46,2){3}{210}{270}
\pcline[linewidth=1.5pt,linecolor=lightblue](-2,0)(-2,-1)
\psbezier[linewidth=1.5pt,linestyle=dashed,linecolor=green](1.47,.85)(2.17,-.36)(3,-0.7)(3.46,-0.7)
\pcline[linewidth=1.5pt,linestyle=dashed,linecolor=green](-1.7,0)(-1.7,-1)
}
\rput(-4.23,1.2){
\psbezier[linewidth=1.5pt,linestyle=dashed,linecolor=blue](0,0)(3.6,-2.7)(3,1.5)(0.6,0.4)
\psbezier[linewidth=1.5pt,linestyle=dashed,linecolor=blue](0,0)(-2.4,1.8)(-4,-2.2)(-0.4,-0.4)
}
\rput(-4.23,-0.5){
\psarc[linewidth=1.5pt,linestyle=dashed,linecolor=green](2.23,0){.3}{0}{90}
\psarc[linewidth=1.5pt,linestyle=dashed,linecolor=red](-2.23,0){.3}{90}{180}
}
\rput(-4.23,-0.5){
\psarc[linewidth=1.5pt,linecolor=lightblue](0,0){3.23}{180}{360}
\psarc[linewidth=1.5pt,linecolor=lightblue](0,0){2.23}{180}{360}
\psarc[linewidth=1.5pt,linestyle=dashed,linecolor=green](-0.2,0){2.73}{180}{360}
\psarc[linewidth=1.5pt,linestyle=dashed,linecolor=red](0.2,0){2.73}{180}{360}
}
\rput(-7.5,2.5){\makebox(0,0){$Z_1$}}
\rput(-1,2.5){\makebox(0,0){$Z_2$}}
\rput(-4.26,-2){\makebox(0,0){$Z_3$}}
%%%%
%
\pcline[linewidth=2pt]{->}(0,0)(1,0)
\rput(5,0){
\psarc[linewidth=1.5pt,linecolor=lightblue](0,0.5){3}{0}{180}
\psarc[linewidth=1.5pt,linecolor=lightblue](0,0.5){2}{0}{180}
\psarc[linewidth=1.5pt,linecolor=lightblue](0,-0.5){3}{180}{360}
\psarc[linewidth=1.5pt,linecolor=lightblue](0,-0.5){2}{180}{360}
\pcline[linewidth=1.5pt,linecolor=lightblue](-2,0.5)(2,0.5)
\pcline[linewidth=1.5pt,linecolor=lightblue](-2,-0.5)(2,-0.5)
\pcline[linewidth=1.5pt,linecolor=lightblue](-3,0.5)(-3,-0.5)
\pcline[linewidth=1.5pt,linecolor=lightblue](3,0.5)(3,-0.5)
\psarc[linewidth=1.5pt,linestyle=dashed,linecolor=green](0,0.5){2.7}{0}{180}
\psarc[linewidth=1.5pt,linestyle=dashed,linecolor=green](0,-0.5){2.7}{180}{360}
\pcline[linewidth=1.5pt,linestyle=dashed,linecolor=green](2.7,0.5)(2.7,-0.5)
\pcline[linewidth=1.5pt,linestyle=dashed,linecolor=green](-2.7,0.5)(-2.7,-0.5)
\psarc[linewidth=1.5pt,linestyle=dashed,linecolor=blue](0,0.5){2.3}{0}{180}
\psarc[linewidth=1.5pt,linestyle=dashed,linecolor=blue](-2,0.5){.3}{180}{270}
\psarc[linewidth=1.5pt,linestyle=dashed,linecolor=blue](2,0.5){.3}{270}{360}
\pcline[linewidth=1.5pt,linestyle=dashed,linecolor=blue](-2,0.2)(2,0.2)
\psarc[linewidth=1.5pt,linestyle=dashed,linecolor=red](0,-0.5){2.3}{180}{360}
\psarc[linewidth=1.5pt,linestyle=dashed,linecolor=red](-2,-0.5){.3}{90}{180}
\psarc[linewidth=1.5pt,linestyle=dashed,linecolor=red](2,-0.5){.3}{0}{90}
\pcline[linewidth=1.5pt,linestyle=dashed,linecolor=red](-2,-0.2)(2,-0.2)
}
\rput(2,2.5){\makebox(0,0){$X_1$}}
\rput(5,1){\makebox(0,0){$X_2$}}
\rput(8,-2.5){\makebox(0,0){$X_3$}}
\end{pspicture}
}
\caption{\small The transformation between geodesics in Fig.~\ref{fi:correspondence} depicted for the spines
$\Gamma_{1,1}$ and $\Gamma_{0,3}$.}
\label{fi:fatgraph}
\end{center}
\end{figure}

\end{example}

\section{Conclusion}\label{se:con}

Theorems \ref{thm:JordanAn} and \ref{thm:JordanCFP} characterise the Stokes
matrices arising in the
Teichm\"uller theory of a Riemann sphere with one hole and $n$ orbifold points and
of a Riemann surface of genus $g$ with one or two holes respectively.

In section \ref{se:bondal},   we have seen that all Strokes matrices belonging to
the degenerated symplectic leaves $\mathcal L_{\mathcal A_n}$ and
$\mathcal L_{CFP}$ can be parameterised in terms of complex
coordinates $Z_1,\dots,Z_n,Y_1,\dots,Y_k$, where $k=n-3$ in the
${\mathcal A_n}$ case and $k=2 n-6$ in the $CFP$ case.

In order to characterise the Frobenius Manifold structure
corresponding to these degenerated symplectic leaves one possible
strategy is to determine the solution $V(u_1,\dots,u_n)$ of the
isomonodromic deformation equation (\ref{15}) and then to use
Dubrovin's isomonodromicity theorem part III in \cite{Dub7}  to
reconstruct the metric, the flat coordinates, the pre--potential and
the structure constants of the Frobenius manifold. Unfortunately at
the moment this strategy fails at the very first step, i.e. we are
unable to determine  $V(u_1,\dots,u_n)$, even in the simplest case,
i.e. for $n=3$.  We are going to explain what happens in this case
in the next subsection and then in subsection \ref{suse:higher} we
will say a few words about the case of $n>3$.

\subsection{Case $n=3$}\label{se:PVI}

For $n=3$ we deal only with  $\mathcal L_{\mathcal A_3}$  (the CFP
case for $n=3$ is completely equivalent to this one up to doubling
of the shear coordinates). The geodesics $G_{ij}$ are given by the
following formula in which we use cyclic notation:
\be\label{eq:GZ}
G_{i ,i+1}=  e^{Z_i+Z_{i+1}} + e^{Z_i-Z_{i+1}} + e^{-Z_i-Z_{i+1}},
\ee
which for $Z_1,Z_2,Z_3\in\mathbb R$ are strictly bigger than $2$. In
this case all Stokes matrices in the generic symplectic leaves can
be parameterised in terms of the complexified shear coordinates
$Z_1,Z_2,Z_3$, simply imposing
$$
S=\left(\begin{array}{ccc}1&G_{1,2}&G_{3,1}\\
0&1&G_{2,3}\\ 0&0&1\\
\end{array}\right),
$$
where now $G_{i,j}$ are given by (\ref{eq:GZ}) with complex $Z_1,Z_2,Z_3$.

For generic values of the central element $p=Z_1+Z_2+Z_3$, the
Jordan normal form $J_0$ of the monodromy around $0$ of system
(\ref{irreg}) is actually diagonal,
$$
J_0=\left(
\begin{array}{ccc}
e^{p}&0&0\\
0&e^{-p}&0\\
0&0&1\\ \end{array}
\right),
$$
so that the matrix $V$ is non resonant and the Stokes matrix $S$
determines uniquely the local solutions $V(u_1,u_2,u_3)$ of the
isomonodromic deformation equations  (\ref{15}). In this case the
isomonodromic deformation equations reduce to a special case of the
sixth Painlev\'e equation \cite{Dub6}
\begin{eqnarray}\label{eq:PVI}
\ddot y&=&{1\over2}\left({1\over y}+{1\over y-1}+{1\over y-t}\right) \dot y^2 -
\left({1\over t}+{1\over t-1}+{1\over y-t}\right)\dot y+\nn\\
&+&{y(y-1)(y-t)\over t^2(t-1)^2}\left[\frac{(2\mu-1)^2}{2} +\frac{1}{2} {t(t-1)\over(y-t)^2}\right],
\end{eqnarray}
where $\mu=\frac{Z_1+Z_2+Z_3}{4 i\pi}$,  $t=\frac{u_2-u_1}{u_3-u_1}$ and
$$
y = \frac{t\left(V_{12} V_{23}+\mu V_{13}\right)^2}{(t-1)(\mu + V_{12}^2)+
t\left(V_{12}V_{23}+\mu V_{13}\right)^2},
$$
so that the entries in the Stokes matrix uniquely determine the
local solutions of this special case (\ref{eq:PVI}) of the sixth
Painlev\'e equation.

The generic solutions of the sixth Painlev\'e equation are {\it
irreducible}\/ transcendental functions, i.e. they  cannot be
expressed via elementary or classical transcendental functions by
simple operations. Of course some special solutions may be
reducible:  indeed all algebraic solutions of (\ref{eq:PVI}) were
classified in \cite{DM} and \cite{mazz1}, and the so called {\it
classical solutions,}\/ solutions that can be expressed in terms of
hypergeometric functions were classified in \cite{W}.  However, in
the geometric case, i.e. for $Z_1,Z_2,Z_3\in\mathbb R$, the
solutions are certainly irreducible: in \cite{DM} and \cite{mazz1}
it was proved that in order to have algebraic solutions, a necessary
condition is that $|S_{i,j}|<2$, which is clearly violated in the
geometric case. Moreover, using the results of \cite{mazz}, it is
rather straightforward to prove that these solutions are never of
hypergeometric type.

Another nasty surprise is given by looking at the asymptotic
behaviour of the geometric solutions near the critical points.
Indeed most  PVI solutions have asymptotic behaviour of algebraic
type, namely given $\sigma_i$, $i=1,2,3$  complex numbers such that
$$
2 \sin\frac{\pi\sigma_i}{2}= S_{jk}, \quad i\neq j,k,\quad\hbox{and}
\quad \Re(\sigma_i)\in]0,1[,
$$
the corresponding PVI solution has the following asymptotic behaviours of algebraic type \cite{jimbo}:
$$
y(t)\sim\left\{\begin{array}{lc}
a_0 t^{1-\sigma_3} (1+\mathcal O(t))&\hbox{for } t\to 0,\\
1-a_1 (1-t)^{1-\sigma_2} (1+\mathcal O(1-t))&\hbox{for } t\to 1,\\
a_\infty t^{\sigma_1}  (1+\mathcal O(\frac{1}{t}))&\hbox{for } t\to \infty.\\
\end{array}\right.
$$
However, for $S_{i,j}=G_{i,j}>2$, we have $\sigma_i=1+ i \nu_i$,
$\nu_i\in\mathbb R$ for all $i=1,2,3$. In this case the asymptotics
are no longer of algebraic type, but become very complicated
\cite{guzzetti1}. For example near $0$ we have:
$$
y(t)\sim \frac{1}{\sin^2\left(\frac{\nu}{2}\log(x) + \phi +\frac{\nu}{2}\ F_1(x)/F(x) \right)},
$$
where $\phi$ is a phase parameter and $F(x), F_1(x)$ are the two
Jacobi elliptic integrals. This makes all asymptotic computations of
$V$ and of the metric, the flat coordinates, the pre--potential and
the structure constants of the Frobenius manifold extremely involved
if not impossible.

\subsection{Higher $n$}\label{suse:higher}

First observe that the discrepancy $d$ between the dimension of the
generic symplectic leaves and the dimension of the leaves $\mathcal
L_{\mathcal A_n}$ and $\mathcal L_{CFP}$ is given by:
$$
d_{\mathcal A_n}:=\dim(\mathcal L_{generic})-\dim(\mathcal L_{\mathcal A_n})=
\left\{\begin{array}{lc}
\frac{1}{2}(n-3)^2&\hbox{for } n \hbox{ odd}\\
\frac{1}{2}(n-2)(n-4)&\hbox{for } n \hbox{ even}\\
\end{array}\right.
$$
$$
d_{CFP}:=\dim(\mathcal L_{generic})-\dim(\mathcal L_{CFP})=
\left\{\begin{array}{lc}
\frac{1}{2}(n-3)(n-5)&\hbox{for } n \hbox{ odd}\\
\frac{1}{2}(n-4)^2&\hbox{for } n \hbox{ even}\\
\end{array}\right.
$$
we see that for $n=3,4$ the leaves  $\mathcal L_{\mathcal A_n}$ are
generic, while for $n=4,5$ the leaves  $\mathcal L_{CFP}$ are
generic.  As we have observed above, this fact is at the root of why
we can't actually solve the isomonodromic deformation equations
(\ref{15}): for small $n$ we deal with "generic solutions"  which,
as we have seen above, are irreducible transcendental functions.

For $n>5$ the discrepancy $d$ between the dimension of the generic
symplectic leaves and the dimension of  the leaves  $\mathcal
L_{\mathcal A_n}$ and  $\mathcal L_{CFP}$  is non zero. In terms of
solutions $V(u_1,\dots,u_n)$ of the isomonodromic deformation
equation (\ref{15}), this means that the matrix function
$V(u_1,\dots,u_n)$ satisfies extra $d$ independent equations. These
are algebraic equations that can be obtained by observing that as
soon as $n$ is large enough, $J_0$ has a block diagonal form in
which one block is the minus identity. This means that $V$ is
resonant, and in principle we should have
$$
J_0=\exp(2\pi i\mu)\exp(2\pi i R),
$$
where $R$ is a nilpotent matrix satisfying (\ref{lev1}) which can be
recursively determined in terms of the entries of $V$. When a minus
identity diagonal block appears, all off diagonal entries
corresponding to that diagonal block must be zero, leading to extra
equations for $V$. For example in the $\mathcal A_n$ case, for $n$
even we have $n-2$ eigenvalues equal to $-1$, so we should expect
$R$ to have $\frac{(n-2)(n-3)}{2}$ off diagonal entries. Since on
our degenerated symplectic leaf only one of those in non zero, we
expect $\frac{(n-4)(n-1)}{2}$ equations of which only
$d=\frac{(n-4)^2}{2}$  are independent. Following the same train of
thoughts as in \cite{DM1}, this implies that  the solution
$V(u_1,\dots,u_n)$ of the isomonodromic deformation equation
(\ref{15}) corresponding to the degenerated symplectic leaves can be
in fact reduced to the Garnier system both in the $\mathcal A_n$ and
in the $CFP$ case. Work on this reduction is still in progress.

%\pdfbookmark[1]{References}{ref}

\end{document}